\documentclass[12pt]{article}

\usepackage{amsthm}
\usepackage{amsmath,amssymb,amsfonts,bm,amsbsy,bbm}
\usepackage{epsfig}
\usepackage{graphicx,rotating,lscape}
\usepackage{hyperref}
\usepackage{verbatim}
\usepackage{natbib}
\usepackage{multirow,color}
\usepackage{geometry}
\usepackage{setspace}
\usepackage{subfigure}
\usepackage{pdflscape}
\usepackage{longtable}
\usepackage{authblk}
\usepackage{algorithm}
\usepackage{algpseudocode}

\renewcommand{\theequation}{\arabic{equation}}
\newtheorem{lemma}{Lemma}       
\newtheorem{theorem}{Theorem}      
\newtheorem{corollary}{Corollary}

\newcommand{\pre}{\textnormal{pre}}
\newcommand{\post}{\textnormal{post}}
\newcommand{\deb}{\textnormal{deb}}
\newcommand{\obs}{\textnormal{obs}}

\DeclareMathOperator{\sign}{sign}
\DeclareMathOperator*{\argmin}{arg\,min}
\DeclareMathOperator*{\argmax}{arg\,max}

\def\bse{\begin{eqnarray*}}
\def\ese{\end{eqnarray*}}
\def\be{\begin{eqnarray}}
\def\ee{\end{eqnarray}}
\def\bsq{\begin{equation*}}
\def\esq{\end{equation*}}
\def\bq{\begin{equation}}
\def\eq{\end{equation}}

\def\trans{^{\rm T}}
\def\wh{\widehat}
\def\wt{\widetilde}
\def\wb{\overline}
\def\ind{\mathbbm{1}}
\def\th{^{th}}
\def\kth{^{(k)}}

\def\sumi{\sum_{i=1}^n}

\def\boxit#1{\vbox{\hrule\hbox{\vrule\kern6pt\vbox{\kern6pt#1\kern6pt}\kern6pt\vrule}\hrule}}

\def\bb{{\boldsymbol{\beta}}}
\def\ett{{\boldsymbol{\eta}}}

\def\rr{{\boldsymbol{\gamma}}}

\def\tth{{\boldsymbol{\theta}}}

\def\xxi{{\boldsymbol{\xi}}}

\def\a{{\bf a}}

\def\b{{\bf b}}

\def\G{{\bf G}}

\def\H{{\bf H}}

\def\S{{\bf S}}

\def\v{{\boldsymbol{v}}}

\def\W{{\bf W}}
\def\Z{{\bf Z}}
\def\z{{\bf z}}

\def\ml{{\mathcal{L}}}

\def\pr{\hbox{pr}}

\def\MCID{\hbox{MCID}}
\def\iMCID{\hbox{iMCID}}

\def\mcid{\hbox{MCID}}

\begin{document}
	\thispagestyle{empty}
	\allowdisplaybreaks
	\renewcommand {\thepage}{}
	\include{titre}
	\pagenumbering{arabic}

	\begin{center}
		{\Large{\bf A Statistical Inference Framework for the Minimal Clinically Important Difference}}
		\vskip 2mm {\bf Zehua Zhou, Leslie J. Bisson and Jiwei Zhao}
		\vskip 1mm
		Global Biometrics \& Data Sciences, Bristol Myers Squibb, New Jersey, USA
		
		Department of Orthopaedics, State University of New York at Buffalo, New York, USA
		
		Department of Biostatistics and Medical Informatics, University of Wisconsin-Madison, Wisconsin, USA
	\end{center}

\begin{abstract}
In clinical research, the effect of a treatment or intervention is widely assessed through clinical importance, instead of statistical significance.
	In this paper, we propose a principled statistical inference framework to learn the minimal clinically important difference (MCID), a vital concept in assessing clinical importance.
	We formulate the scientific question into a statistical learning problem, develop an efficient algorithm for parameter estimation, and establish the asymptotic theory for the proposed estimator.
	We conduct comprehensive simulation studies to examine the finite sample performance of the proposed method.
	We also re-analyze the ChAMP (Chondral Lesions And Meniscus Procedures) trial, where the primary outcome is the patient-reported pain score change and the ultimate goal is to determine whether there exists a significant difference in post-operative knee pain between patients undergoing debridement versus observation of chondral lesions during the surgery.
	Some previous analysis of this trial exhibited that the effect of debriding the chondral lesions does not reach a statistical significance. Our analysis reinforces this conclusion in that the effect of debriding the chondral lesions is not only statistically non-significant, but also clinically un-important.

	
\end{abstract}

\newpage

\section{Introduction}\label{sec:intro}

In clinical research, the effect of a targeted therapeutic treatment or intervention is widely assessed through statistical significance.
The p-value with a subjective threshold, say, 5\%, has been extensively used in the literature to assess statistical significance versus non-significance, for either objective endpoints or patient-reported outcomes (PRO).

However, the limitations as well as some misunderstandings of statistical significance \citep{sterne2001sifting, wasserstein2016asa, amrhein2019retire, wasserstein2019moving} have been documented in the literature over the years.
Conclusions solely based on statistical significance could be misleading in many occasions.
It is well known that statistical significance only signifies the existence of treatment effect, regardless of its effect size.
Further, statistical significance is linked
to the sample size. Given a large enough sample, statistical significance between subgroups may occur with very small differences that are clinically meaningless.
Therefore, a better assessing instrument for the clinical importance is desired in applications.

\subsection{What is Minimal Clinically Important Difference}\label{sec:introMCID}

Ideally, an appropriate clinical interpretation of the treatment effect should consider not only statistical significance but also whether the observed change is indeed meaningful to the patients.
A natural approach to do so is to determine the minimal clinically important difference (\mcid), a concept firstly introduced in 1989 \citep{jaeschke1989measurement}. In their seminal work, the authors defined the \mcid~as \emph{the smallest difference in score in the domain of interest which patients perceive as beneficial and which would mandate, in the absence of troublesome side effects and excessive cost, a change in the patient's management}.
For example, in the ChAMP trial that will be introduced in Section~\ref{sec:data}, it is the smallest WOMAC pain score change such that the corresponding improvement and beyond can be claimed as clinically important. 
The MCID is a patient-centered concept, capturing both the magnitude of the improvement and also the value patients place on the change \citep{mcglothlin2014minimal}.
Since its debut in 1989, the research on MCID has gradually attained its booming popularity in the literature.
For more details of this topic, we recommend a recent review paper \citep{mcglothlin2014minimal} published at \emph{JAMA Guide to Statistics and Methods}.

Current methods of determining $\MCID$ can be roughly clustered into three categories: distribution-based, opinion-based, and anchor-based \citep{lassere2001foundations, erdogan2016minimal, angst2017minimal, jayadevappa2017minimal}.
The distribution-based methods compare the change of the outcome with some variability measure of the sample,
such as the standard deviation (SD), the standard error of measurement (SEM) or the effect size (ES).
For example, the previous work simply defined $1$ SEM as the value of MCID in heart failure \citep{wyrwich1999linking} and respiratory disease patients \citep{wyrwich1999further}, \cite{samsa1999determining} proposed to use $0.2$ SD of baseline outcome score as \mcid, and the review paper \citep{norman2003interpretation} summarized a variety of studies and identified $0.5$ SD as a threshold value of the important change. Clearly, even for the same disease and the same PRO instrument, different studies might propose different values of MCID.
Additionally, according to the \emph{JAMA} review paper \citep{mcglothlin2014minimal}, \emph{distribution-based methods are not derived from individual patients, hence they probably should not be used to determine an MCID}.

Opinion-based method, also called the Delphi method, determines the MCID by
seeking the opinions from a panel of clinical experts who achieve consensus regarding their best bet on $\MCID$ \citep{bellamy2001towards}. This is a somewhat subjective method, and as appointed \citep{mcglothlin2014minimal}, \emph{in many settings, expert opinion may not be a valid and reliable way to determine what is important to patients}.

The third method is anchor based, which
essentially matches the change of the outcome to an external PRO, usually expressed as a binary variable derived from the anchor question.
The ChAMP trial used the second question of SF-36 survey \citep{angst2001smallest} as the anchor question: \emph{Overall do you feel an improvement after receiving the surgery/treatment?}.
The answer to the anchor question facilitates us to identify the threshold value of the outcome change through diagnostic medicine tools such as the receiver operating characteristic curve \citep{stratford1998sensitivity, riddle1998sensitivity}.
The anchor-based methods have the potential to lead to a rigorous treatment for the MCID \citep{hedayat2015minimum,zhou2019estimation}, which are indeed the base of the proposed method in this paper. On the other hand, we would like to stress that, our work is motivated from resolving two limitations of the existing anchor-based methods, with details articulated below.

\subsection{Our Contribution}\label{sec:introcontribution}

The first limitation of the existing methods stems from the difficulty of defining MCID for different subgroups of patients or even each single individual.
This point was well noted before \citep{jaeschke1989measurement}, where the authors stated \emph{this information will be useful in interpreting questionnaire
	scores, both in individuals and in groups of patients participating
	in controlled trials}. This was also documented in \cite{mcglothlin2014minimal} that \emph{determination of the MCID should consider different
	thresholds in different subsets of the population}. That means, it is preferred to derive the MCID not only among the whole population (pMCID), but also for a certain subgroup defined by the individuals' clinical profile (iMCID). In general, it is believed that the MCID can vary according to the characteristics of patients, especially pertaining to the severity of disease and baseline health status \citep{stratford1998sensitivity, riddle1998sensitivity}, or some other forms of patient variation, such as age, socioeconomic status, and education \citep{lauridsen2006responsiveness}.

The other issue is how to quantify the uncertainty of the estimate of the MCID. Only a point estimate itself cannot provide the knowledge of how precise the estimate is. In principle, one needs to derive its standard error as well as the associated interval estimation, to quantify the uncertainty of the MCID estimates.
\cite{crosby2003defining} pointed out that the statistical inference of $\MCID$ should be considered with respect to populations as well as to individuals. Inference at the population level (for pMCID) could inform the comparisons between different treatments or decisions regarding public policy, whereas the inference at the individual level (for iMCID) is likely to provide suggestions on the individual clinical treatment decisions, which fits the ideas of personalized medicine \citep{kee2020scientific}.
Ideally, it could solve this problem to derive the asymptotic representation of the point estimator based on normal approximation; however, in this paper, because of the use of the zero-one loss and its surrogate, it is not a trivial task to establish the asymptotic normality of the estimator.

Below we briefly summarize the novel contributions we make.
Firstly, we follow \cite{zhou2019estimation} to formulate the estimation of MCID into a statistical learning problem with the zero-one loss; however, quite different from \cite{zhou2019estimation}, we concentrate on the linear form of iMCID due to its transparency and ease of interpretation. 
Second, since the optimization with the zero-one loss is empirically infeasible and a popular convex surrogate might be Fisher inconsistent, we propose to apply a nonconvex surrogate for algorithmic implementation.
We identify a class of such nonconvex surrogates, which are different from those in the existing literature \citep{hedayat2015minimum, zhou2019estimation}, and figure out some regularity conditions such that the Fisher consistency is still valid.
We implement the popular difference of convex (DC) algorithm to deal with the non-convexity in optimization.
Third, to quantify the uncertainty and perform the inference, we rigorously establish the asymptotic representation for the MCID estimate.
This permits us to derive the closed-form expression of the standard error, and to propose the corresponding interval estimation based on the normal approximation.
We would like to stress that the nonconvex surrogate proposed in this paper is different from that in an existing paper \citep{zhou2019estimation}. Because of the singularity in the asymptotic variance, the surrogate used in \cite{zhou2019estimation}, simply based on linear approximation, becomes problematic for developing asymptotic theory based on which the interval estimation is constructed.
To the contrary, a class of nonconvex surrogates proposed in this paper, is diligently designed bearing in mind the derivation of valid asymptotic variances and the construction of corresponding interval estimates.

The remainder of this paper is organized as follows.
Section~\ref{sec:method} contains the methodology of our work.
Section~\ref{sec:prelim} introduces the whole set-up, pinpoints the main difficulty in this problem, and proposes the key idea to overcome this difficulty.
Next, Section~\ref{sec:alg} details the algorithmic implementation and Section~\ref{sec:theory} establishes the asymptotic theory of the MCID estimate.
A comprehensive simulation study was conducted in Section~\ref{sec:simu}, followed by the analysis results for the ChAMP trial in Section~\ref{sec:data}. The paper is concluded with a discussion in Section~\ref{sec:disc}. All the technical proofs are contained in the online supplementary materials.

\section{Method}\label{sec:method}

\subsection{Problem Formulation}\label{sec:prelim}

We first introduce some notation.
Let continuous random variable $X\in R$ be the measurement of patient's outcome change, i.e. $X=X_{\rm{post}}-X_{\rm{pre}}$ where $X_{\rm{pre}}$ and $X_{\rm{post}}$ denote the outcome before and after the treatment, respectively. Let binary variable $Y$ be the PRO of interest where $Y=1$ if the patient reports an improved health condition after receiving the treatment and $Y=-1$ otherwise.
Assume the marginal distribution of $Y$ follows $\pr(Y=1)=\pi$ with $\epsilon\leq\pi\leq1-\epsilon$ and $\epsilon$ is a small positive quantity. Let $p_{+}(x)$ and $p_{-}(x)$ represent the conditional probability density function (pdf) of $X$ given $Y=1$ and $Y=-1$, respectively.
We denote $\tau$ as the value of the population MCID (pMCID) and $\tau^0$ its truth.

If we recast the relation between $Y$ and $X$ as a classification problem, the optimal cutoff value $\tau^0$, using the Youden's index, is the maximizer of
\be\label{eq:youden}
\pr(X\geq \tau\mid Y=1) + \pr(X<\tau\mid Y=-1).
\ee
One can derive that, solving (\ref{eq:youden}) is equivalent to finding the solution of
\bq
\argmin_\tau ~ E[w(Y)L_{01}\{Y(X-\tau)\}],
\label{objfun1}
\eq
where $L_{01}(u)=\frac12\{1-\sign(u)\}$, $\sign(u)=1$ if $u\geq0$ and $-1$ otherwise, $w(1)=1/\pi$ and $w(-1)=1/(1-\pi)$. Note that the prior work \citep{hedayat2015minimum} defined MCID using a similar idea; however, their definition is not compatible with an imbalanced $Y$ \citep{zhou2019estimation}.

The pMCID itself is usually not informative enough \citep{wells2001minimal}, as it only informs the comparison between different treatments or decisions.
Various factors such as the patient's baseline status and demographic variables could have impacts on the PRO change and further affect $\tau^0$. For example, a shoulder pain reduction study showed that, due to the higher expectations for complete recovery, the healthier patients with mild pain at baseline often deemed greater pain reduction as ``beneficial'' or ``meaningful'' than the ones who suffered from chronic disease \citep{heald1997shoulder}. Thus, patients with different clinical profiles may have different expectations of the important changes, thereby having different values of $\tau^0$.
Hence, it is desirable to estimate the individual MCID (iMCID) for some identified subgroups or individuals based on prognostic variables such as treatment, age and sex.
Let $\Z$ represent the patient's prognostic profile, then the $\iMCID$ can be expressed as a function of $\Z$,  $\tau(\Z)$. Similar to (\ref{objfun1}), the true value of iMCID, $\tau^0(\Z)$, is formulated as
\bq
\argmin_{\tau(\z)} ~ E\left[w(Y)L_{01}\{Y(X-\tau(\Z))\}|\Z=\z\right].
\label{objfun2}
\eq
In clinical practice, the $\iMCID$ with a simple structure, such as linear, is preferred due to its transparency and convenience for interpretation. In this paper, we consider the iMCID with a linear structure, $\tau(\Z)=\bb\trans\Z$, where $\Z\in\mathcal{R}^p$ with scalar $1$ as its first coordinate and $\bb$ with the same dimension.
Define $p_{+,\z}(x)$ ($p_{-,\z}(x)$) as the conditional pdf of $X$ given $Y=1$ ($Y=-1$) and $\Z=\z$.
We have the following result.
\begin{lemma}
	Assume $f(t)$ is a generic pdf satisfying, $f(t)$ is symmetric around 0, i.e., $f(-t)=f(t)$, and $f(t)$ is monotonically increasing when $t<0$ and decreasing when $t>0$.
	Let $p_{+,\z}(x)=f(x-\a\trans\z)$ and $p_{-,\z}(x)= f(x-\b\trans\z)$ where $\a,\b\in\mathcal{R}^{p}$ exist such that $\a\succ\b$, that is, $a_i>b_i$ for $i=1,\cdots,p$. Then $\tau^0(\z)$ has a unique form $\tau^0(\z)=(\bb^0)\trans\z$ with $\bb^0=(\a+\b)/2$.
	\label{lemma1}
\end{lemma}

Lemma \ref{lemma1} essentially pinpoints some conditions under which $\tau^0(\Z)$ is uniquely well-defined and has an explicit linear form $\tau^0(\Z)=(\bb^0)\trans\Z$. Hence, $\bb^0$ is also uniquely well-defined.
In the remainder of the paper, our interest becomes the estimation and inference of $\bb$, with the following objective function
\bq
\argmin_{\bb} ~ E\left[w(Y)L_{01}\{Y(X-\bb\trans\Z)\}\right].
\label{objfun3}
\eq
Meanwhile, one would easily identify many density functions $f(\cdot)$ that satisfy the conditions in Lemma \ref{lemma1}, e.g. the Normal distribution, the Student's t-distribution, the Laplace distribution, the Cauchy distribution, and the Logistic distribution.

The direct minimization of (\ref{objfun3}) is a nondeterministic polynomial-time (NP) hard problem \citep{natarajan1995sparse} because of the zero-one loss; hence empirically infeasible.
In the literature, one usually replaces the zero-one loss with a convex surrogate loss such as the hinge loss or the logistic loss, which can be shown to be Fisher consistent for traditional classification problems where the boundary function is unrestricted.
However, the boundary function in our problem is in the format of $X-\bb\trans\Z$.
Under such a case, a simplistic convex surrogate cannot be guaranteed to be Fisher consistent \citep{hedayat2015minimum}.

To resolve this issue, we propose a novel family $\{L(u)\}$ of nonconvex surrogate loss functions, with $L(u)$ defined as
\[ L(u)=\begin{cases}
1, & u\leq 0, \\
l(u), & 0<u\leq\frac12, \\
1-l(1-u), & \frac12<u\leq1, \\
0, & u>1,
\end{cases}
\]
where $l(u)$ only needs to satisfy
\begin{enumerate}
	\item $l(u)$ is concave for $0<u\leq1/2$ with $l(0)$=1 and $l(1/2)=1/2$;
	\item $l^\prime(0)=0$ and $l^\prime(1/2)=-k$ with $k>0$;
	\item $l^{\prime\prime}(u)$ is symmetric around $1/2$, that is, $l^{\prime\prime}(u)=l^{\prime\prime}(1/2-u)$, and $|l^{\prime\prime}(u)|\leq M$,
\end{enumerate}
where $l^\prime(u)$ and $l^{\prime\prime}(u)$ are the first and second order derivatives of $l(u)$.
The idea of imposing conditions on $l'(u)$ and $l''(u)$ is for the sake of estimating the standard error of the iMCID estimate. This point will become self-explanatory in our asymptotic theory derivation. Clearly, a linear $l(u)$, such that $l(u)=1-u$ used in prior work \citep{hedayat2015minimum}, does not belong to our family.
One simple example in our family is a quadratic function $l(u)=1-2u^2$ with the corresponding $k=2$, which will be used as the exemplar in our numerical studies. We denote it as $L^q(u)$.

We further introduce $L_{\delta}(u)=L(u/\delta)$, where $\delta>0$ can be regarded as a tuning parameter. One can easily verify, for any $u$, $L_{\delta}(u)\rightarrow L_{01}(u)$, when $\delta\rightarrow0$.
With $L_{01}(u)$ in (\ref{objfun3}) replaced by $L_\delta(u)$, our objective now becomes to minimize
\be
\ml(\bb) = E[w(Y)L_\delta\{Y(X-\bb\trans\Z)\}].
\label{objfun4}
\ee
Denote the minimizer of (\ref{objfun4}) as $\bb^*$.
Generally, $\bb^*$ is not the same as $\bb^0$, the uniquely well-defined minimizer of (\ref{objfun3}).
Our next result manifests that, given the conditions in Lemma~\ref{lemma1}, $\bb^*$ is indeed exactly the same as $\bb^0$.
\begin{lemma}\label{theorem1}
 Under the conditions given in Lemma \ref{lemma1}, $\bb^*$ is uniquely well-defined and $\bb^*=\bb^0$.
\end{lemma}
Lemma~\ref{theorem1} evinces that, using the nonconvex surrogate loss $L_\delta(\cdot)$ in our problem is not merely \emph{approximately} Fisher consistent, but actually \emph{exactly} Fisher consistent.
This is a pivotal result for our theoretical derivation in the next section. With
Lemma~\ref{theorem1}, there is no discrepancy of centralizing our iMCID estimator to $\bb^0$ or to $\bb^*$.

\subsection{Algorithm}\label{sec:alg}

Suppose we observe the independent and identically distributed data $\{(x_i,y_i,\z_i,w_i), i=1,\cdots,n\}$ with $w_i=w(y_i)$.
Let $\bb_{-1}=\bb\setminus\beta_1$, where $\beta_1$ is the first element in $\bb$, i.e., the intercept term.
By introducing an $L_2$ penalty on $\bb_{-1}$ to avoid the model over-fitting, we minimize the empirical version of (\ref{objfun4}) as
\bq
\ml_\lambda(\bb)=\frac1n\sumi w_i L_\delta\{y_i(x_i-\bb\trans\z_i)\}+\frac\lambda2||\bb_{-1}||^2,
\label{objfun5}
\eq
where $\lambda>0$ is the tuning parameter.
We denote its minimizer as $\wh\bb_\lambda$.

To tackle the non-convexity issue in $\ml_\lambda(\bb)$, we adopt the popular difference of convex (DC) algorithm \citep{le1997solving}.
Firstly, we decompose $L_\delta(u)$ as the difference of two convex functions, i.e.,
$L_\delta(u) = L_{\delta1}(u)-L_{\delta2}(u)$, where $L_{\delta1}(u)=\{-k(u/\delta-1/2)+1/2\}\ind_{\{u\leq\delta/2\}}+\{1-l(1-u/\delta)\}\ind_{\{\delta/2<u\leq\delta\}}$ and $L_{\delta2}(u)=\{-k(u/\delta-1/2)-1/2\}\ind_{\{u\leq0\}}+\{-k(u/\delta-1/2)+1/2-l(u/\delta)\}\ind_{\{0<u\leq\delta/2\}}$.
In particular, for the quadratic surrogate loss $L_\delta^q (u) = L^q(u/\delta)$, we write its decomposition as $L^q_\delta(u) = L^q_{\delta1}(u) - L^q_{\delta2}(u)$.

Then, for $\ml_\lambda(\bb)$ in (\ref{objfun5}), we have $\ml_\lambda(\bb)=\ml_{\lambda1}(\bb)-\ml_{\lambda2}(\bb)$, and
\bse
\ml_{\lambda1}(\bb)&=&\frac1n\sumi w_i\left[-k\left(\frac{\xi_i}\delta-\frac12\right)\ind_{\{\xi_i\leq\frac\delta2\}}+\left\{1-l\left(1-\frac{\xi_i}\delta\right)\right\}\ind_{\{\frac\delta2<\xi_i\leq\delta\}}\right]+\frac\lambda2\bb_{-1}\trans\bb_{-1}, \\
\ml_{\lambda2}(\bb)&=&\frac1n\sumi w_i\left[\left\{-k\left(\frac{\xi_i}\delta-\frac12\right)-\frac12\right\}\ind_{\{\xi_i\leq0\}}+\left\{-k\left(\frac{\xi_i}\delta-\frac12\right)+\frac12-l\left(\frac{\xi_i}\delta\right)\right\}\ind_{\{0<\xi_i\leq\frac\delta2\}}\right],
\ese
where $\xi_i=\xi_i(\bb)=y_i(x_i-\bb\trans\z_i)$.
The $(k+1)\th$ step of the DC algorithm is to minimize
\bq
h(\bb)=\ml_{\lambda1}(\bb)-\sumi\xi_i\frac{\partial \ml_{\lambda2}(\bb)}{\partial\xi_i}\bigg\vert_{\bb=\wh\bb\kth},
\label{objfun6}
\eq
where $\frac{\partial \ml_{\lambda2}(\bb)}{\partial\xi_i}=-\frac1{n\delta} w_i \{k\ind_{\{\xi_i\leq0\}}+(k+l^\prime(\xi_i/\delta))\ind_{\{0<\xi_i\leq\frac\delta2\}}\}$
is the subgradient of $\ml_{\lambda2}(\bb)$ with respect to $\xi_i$.
Any convex optimization method can be used to minimize (\ref{objfun6}).
In our numerical studies, we adopt the
limited-memory Broyden-Fletcher-Goldfarb-Shanno (L-BFGS) algorithm \citep{nocedal1980updating}, and we use cross validation to choose the tuning parameter $\lambda$. We summarize the whole procedure below.
\begin{algorithm}
	\caption{DC algorithm for calculating $\wh\bb_\lambda$}
	\begin{algorithmic}
		\State \textbf{Input} Initial value $\wh\bb^{(0)}$, a small positive value $\epsilon_\bb>0$, tuning parameter $\lambda>0$, $\delta>0$
		\Repeat
		\State Update $\wh\xi_i\kth=\xi_i(\wh\bb\kth)=y_i(x_i-\wh\bb^{(k)\rm T}\z_i)$
		\State Update $\frac{\partial \ml_{\lambda2}(\bb)}{\partial\xi_i}$ with $\xi_i=\wh\xi_i\kth$
		\State Calculate $\wh\bb^{(k+1)}$ by minimizing $h(\bb)=\ml_{\lambda1}(\bb)-\sumi\xi_i\frac{\partial \ml_{\lambda2}(\bb)}{\partial\xi_i}\bigg\vert_{\bb=\wh\bb\kth}$
		\Until $h(\wh\bb\kth)-h(\wh\bb^{(k+1)})<\epsilon_\bb$\newline
		\Return $\wh\bb_\lambda = \wh\bb^{(k+1)}$
	\end{algorithmic}
\end{algorithm}

\subsection{Theory}\label{sec:theory}

Now, we derive the asymptotic representation $\sqrt{n}(\wh\bb_\lambda-\bb^0)$. This will guide the standard error estimation of $\wh\bb_\lambda$.
Let $\xi=\xi(\bb)=Y(X-\bb\trans\Z)$.
Then $\xxi^\prime(\bb)=-Y\Z$, a $p$-dimensional vector, and
$\xxi^{\prime\prime}(\bb)=\bf{0}$, a $p\times p$ zero matrix.
For $\ml(\bb)$ defined in (\ref{objfun4}), we define
$\S(\bb)=(S(\bb)_j)$ as the $p$-dimensional gradient vector given by
\be
\S(\bb)&= & E\left\{w(Y)L_{\delta1}^\prime\left(\xi(\bb)\right)\xxi^\prime(\bb)\right\}-E\left\{w(Y)L_{\delta2}^\prime\left(\xi(\bb)\right)\xxi^\prime(\bb)\right\} \nonumber \\
&= & -\delta^{-1}E\left[\left\{l^\prime\left(\xi(\bb)/\delta\right)\ind_{\left\{0<\xi(\bb)\leq \frac\delta2\right\}}+l^\prime\left(1-\xi(\bb)/\delta\right)\ind_{\left\{\frac\delta2<\xi(\bb)\leq\delta\right\}}\right\}w(Y)Y\Z \right],
\ee
where $L_{\delta1}^\prime(u) = \delta^{-1}l^\prime(1-u/\delta) \ind_{\{\delta/2<u\leq\delta\}} - \delta^{-1}k \ind_{\{u\leq\delta/2\}}$ and $L_{\delta2}^\prime(u) = -\delta^{-1}k \ind_{\{u\leq0\}} - \{\delta^{-1}k+\delta^{-1}l^\prime(u/\delta)\} \ind_{\{0<u\leq\delta/2\}}$,
and $\H(\bb)=(H(\bb)_{jk})$ as the $p\times p$ Hessian matrix
\be
\H(\bb)&=&E\left\{w(Y)L_{\delta1}^{\prime\prime}\left(\xi(\bb)\right)\xxi^\prime(\bb)\xxi^\prime(\bb)\trans\right\}-E\left\{w(Y)L_{\delta2}^{\prime\prime}\left(\xi(\bb)\right)\xxi^\prime(\bb)\xxi^\prime(\bb)\trans\right\} \nonumber \\
&=&\delta^{-2}E\left[ \left\{l^{\prime\prime}\left(\xi(\bb)/\delta\right)\ind_{\left\{0\leq \xi(\bb)\leq\frac\delta2\right\}}-l^{\prime\prime}\left(1-\xi(\bb)/\delta\right)\ind_{\left\{\frac\delta2<\xi(\bb)\leq\delta\right\}}\right\} w(Y)\Z\Z\trans\right],
\ee
where $L_{\delta1}^{\prime\prime}(u) = -\delta^{-2}l^{\prime\prime}(1-u/\delta) \ind_{\{\delta/2<u\leq\delta\}}$ and $L_{\delta2}^{\prime\prime}(u) = -\delta^{-2}l^{\prime\prime}(u/\delta) \ind_{\{0<u\leq\delta/2\}}$.
The following regularity conditions
\begin{itemize}
	\item [(C1)] $\S(\bb)$ is well-defined, $\H(\bb)$ is well-defined and is continuous in $\bb$,
	\item [(C2)] The variable $\Z$ has finite fourth moment, i.e., $E||\Z||^4<\infty$,
	\item [(C3)] $n^{1/2}\lambda\rightarrow0$ and $n^{-1}\delta^{-4}\rightarrow0$,
\end{itemize}
are required to present our main theoretical result:
\begin{theorem}
	Under the regularity conditions (C1)-(C3) and those in Lemmas~\ref{lemma1} and \ref{theorem1}, then $\sqrt{n}(\wh\bb_\lambda-\bb^0)\rightarrow N\left\{\bf{0}, \H^{-1}(\bb^0)\G(\bb^0)\H^{-1}(\bb^0)\right\}$
	in distribution when $n\rightarrow\infty$, where $\G(\bb)=\delta^{-2}E\left[\left\{\left(l^{\prime}\left(\xi(\bb)/\delta\right)\right)^2\ind_{\left\{0<\xi(\bb)\leq \frac\delta2\right\}}+\left(l^{\prime}\left(1-\xi(\bb)/\delta\right)\right)^2\ind_{\left\{\frac\delta2<\xi(\bb)\leq \delta\right\}}\right\}w(Y)^2\Z\Z\trans\right]$.
	\label{theorem2}
\end{theorem}
The most challenging part in the proof of this result is to address the non-differentiability of $\ml_\lambda(\bb)$, much more complicated than the linear support vector machine case \citep{koo2008bahadur}.

Consequently, if we have a new patient with clinical profile $\z_0$, the asymptotic representation of her/his iMCID value $\tau(\z_0;\wh\bb_\lambda) = \wh\bb_\lambda\trans\z_0$ will follow:
\begin{corollary}
	Under the same conditions given by Theorem~\ref{theorem2}, we have $	\sqrt{n}\{\tau(\z_0;\wh\bb_\lambda)-\tau(\z_0;\bb^0)\}\rightarrow N\left\{\boldsymbol{0},\z_0\trans \H^{-1}(\bb^0)\G(\bb^0)\H^{-1}(\bb^0)\z_0\right\}$
	in distribution when $n\rightarrow\infty$.
	\label{corollary1}
\end{corollary}

\section{Simulation Studies}\label{sec:simu}

We conduct some simulation studies to evaluate the finite sample performance of our proposed method, when the data generation process is known.

For pMCID, our purpose is to examine the accuracy and precision of the point estimate, as well as the coverage probability of the interval estimate. We generate the random sample $\{(x_i,y_i): i=1,\cdots,n\}$ as follows.\\
(P1): We first generate $y_i$ from Bern$(\pi)$ with $\pi=0.5$. Then, we generate $x_i$ from $N(a,0.1^2)$ if $y_i=1$ whereas $N(b,0.1^2)$ if $y_i=-1$, where $a=0.2$ and $b=-0.1$. The sample size $n$ ranges from 600 to 1800.

We follow the algorithm and the theory presented above. To be more specific, we use the quadratic surrogate loss $L^q_\delta$ to approximate the zero-one loss, and we set $\delta=0.1$ to optimize the performance of our proposed method. We employ a grid search method with set $\{10^{(s-31)/10}: s=1,\cdots,31\}$ to determine the tuning parameter $\lambda$ by minimizing the Euclidian distance between the true parameter and its estimator. The results are based on 500 simulation replicates.

Table~\ref{tab:sim} summarizes the results for scenario (P1). The conclusions are quite clear. The sample bias of the estimated pMCID is very close to zero. The standard error, computed using the asymptotic variance following Theorem~\ref{theorem2}, is very close to the standard deviation, computed via Monte Carlo approximation. The coverage probability is very close to its nominal level 95\%.

For iMCID, we consider the following three scenarios with random sample $\{(x_i, y_i, \z_i): i=1,\ldots,n\}$, and the sample size $n$ ranges from 100 to 1500.\\
(I1): Let $\z_i=(1,z_{1i})\trans$ and generate $z_{1i}$ from $N(1,0.1^2)$. Generate $y_i$ from Bern$(\pi)$ with $\pi=0.5$. Generate $x_i$ from $N(\a\trans\z_i,0.1^2)$ if $y_i=1$ whereas $N(\b\trans\z_i,0.1^2)$ if $y_i=-1$, where $\a=(0.1,0.55)\trans$ and $\b=(-0.1,0.45)\trans$.\\
(I2): Let $\z_i=(1,z_{1i},z_{2i})\trans$ and generate $(z_{1i},z_{2i})\trans$ from bivariate normal distribution with mean $(1,1)\trans$ and variance $0.1^2I_2$. Generate $y_i$ from Bern$(\pi)$ with $\pi=0.5$. Generate $x_i$ from $N(\a\trans\z_i,1)$ if $y_i=1$ whereas $N(\b\trans\z_i,1)$ if $y_i=-1$, where $\a=(0.05,0.55,1.05)\trans$ and $\b=(-0.05,0.49,0.95)\trans$.\\
(I3): Similar to (I1) except that $z_{1i}$ is generated from $\{1,2\}$ with equal probability.

Note that there are two slightly different approaches to handle scenario (I3). One approach (subgroup analysis) is to separately compute pMCID for each subgroup (determined by the value of $z_{1i}$). This is the same as the treatment-specific MCID in Section~\ref{sec:data}. The other (pooled analysis) is to compute iMCID with a linear structure and then the MCID for each subgroup is derived from the estimated parameter values. This is the same as the Model 1 in Section~\ref{sec:data}. The purpose of scenario (I3) is to examine the differences between these two approaches.

The results for scenarios (I1)-(I3) are summarized in Figures~\ref{fig:sim1}-\ref{fig:sim31}, respectively.
We can see that, from Figures~\ref{fig:sim1} and \ref{fig:sim2}, when the sample size is relatively small, say, 100--200, the coverage probability for the parameter estimates in iMCID is about 80--90\%, a bit below the nominal level 95\%. When we increase the sample size to 500, the coverage probability is well close to the nominal level.
If the covariate $\Z$ is all categorial, from Figure~\ref{fig:sim31},
the subgroup analysis is a bit more accurate than the pooled analysis especially with a limited sample size, but the difference is almost negligible.

\section{Analysis of ChAMP Trial}\label{sec:data}

The ChAMP trial \citep{bisson2015design} is a randomized controlled trial with primary aim to determine whether there exists a significant difference in post-operative knee pain between patients undergoing debridement versus observation of chondral lesions.
In ChAMP, patients with chondral lesions deemed to be significant (Outerbridge Grade II-IV) were randomly allocated to debridement (deb $n=98$) or observation (obs $n=92$). A patient-reported outcome, the WOMAC pain score improvement (the greater the better) from baseline to one year after the surgery, serves as the primary outcome of the trial.
The summary statistics are $26.2\pm16.9$ (mean$\pm$standard deviation) versus $26.6\pm16.6$ (deb versus obs),
and there is no statistical significance that can be detected by using two-sample comparison procedures \citep{bisson2017patient, bisson2018does}.
Thus, there is no statistical evidence to debride the chondral lesions since it would not significantly improve the patient's post-operative quality of life.

In this paper, we re-analyze the ChAMP trial from the perspective of assessing clinical importance.
We denote the baseline WOMAC pain score as $X_\pre$ and WOMAC pain score at one year after surgery as $X_\post$, both of which are scaled from 0 (extreme pain) to 100 (no pain). We encode $X=X_\post - X_\pre$ as the WOMAC pain score improvement, the greater the better.
We adopt the anchor based method that
essentially compares $X$ with an external PRO, which is usually expressed as a binary variable $Y$ derived from the anchor question in the sense that $Y=1$ if patient has a positive response on the change of her/his health status and $Y=-1$ otherwise.
A popular anchor question could be \emph{Overall do you feel an improvement after receiving the surgery/treatment?} The ChAMP trial used the second question of SF-36 survey \citep{angst2001smallest} as the anchor question.
In ChAMP trial, patients' demographic characteristics such as gender and age, and their surgical data including the location and type of meniscal tear, and the grade of chondral lesion were recorded.
We denote patient's prognostic profile
$\Z=\{Z_1, Z_2, Z_3, Z_4, Z_5\}$ with treatment assignment (binary), age (continuous), gender (binary), severity of knee damage at baseline (four-level categorical), and the BMI (continuous).
After excluding missing values, there are $156$ patients with $70$ responded positively on their health status and $86$ negatively.

We first report the results for population level clinical importance, followed by treatment specific clinical importance.
We also show that how to build a general model based on the methods we propose in Section~\ref{sec:method}.

\subsection{Population Level Clinical Importance}\label{sec:datapopulation}

Firstly, a simple application of the method proposed in Section~\ref{sec:method} is not to include variable $\Z$, which results in the population level clinical importance.
Using our proposed method, the estimator for MCID is $18.275$.
This estimator means, regarding the whole population, once the WOMAC pain score improvement is greater than $18.275$, it can be claimed as clinically important.
In our calculations, the tuning parameter $\delta$ is selected as 0.5SD of the WOMAC pain score change (with SD $16.7$).
We alter the $\delta$ value ranging from 0.25SD to 0.75SD as a sensitivity analysis, and the results do not change much.

From above, the (absolute value of) mean difference (MD) between deb and obs is $0.4$ .
The previous literature \citep{johnston2010improving} provided a rule of thumb for the interpretation of MD/MCID, where MD/MCID$>$1 means many patients gain important benefits, 0.5$<$MD/MCID$<$1 stands for an appreciable number of patients gain benefit, and MD/MCID$<$0.5 represents that it is progressively less likely that an appreciable number of patients will achieve important benefits. Clearly, the knowledge of MCID here would not alter the conclusion that there is no need to debride the chondral lesions during the surgery.

\subsection{Treatment Specific Clinical Importance}\label{sec:datatreatment}

Furthermore, to study the MCID in each of the treatment groups, we compute the MCID within deb group (treatment$=$1) and obs group (treatment$=$0), respectively.
To make a comparison, we also employ an ad-hoc method \citep{jaeschke1989measurement} and a method based on the receiver operating characteristic (ROC) curve.
The results are shown in Table~\ref{tab:app11}.
From Table~\ref{tab:app11}, the ad-hoc method and the ROC method give a much larger estimate of MCID, hence might be too conservative.
Also, although the SE from the ad-hoc method looks reasonable, those from the ROC method are too large hence might not be reliable.

Although from Table~\ref{tab:app11} the treatment-specific MCID is informative per se, the comparison between the two groups is not clear. Next, we randomly split the whole data set ($n$=156) into training data ($n$=78) and testing data ($n$=78). In the deb group, we use the training data ($n$=40) to compute the MCID (14.229),
then compare with the WOMAC pain score improvement in the testing data, and find out that 80.0\% (32/40) are greater than this MCID. Similarly, the MCID for the obs group in the training data ($n$=38) is 18.144, and 65.8\% (25/38) in the testing data are greater than this MCID.
Consider the hypothesis
\bse
H_0: p_{\deb}\leq p_{\obs} \mbox{ versus } H_a: p_{\deb}>p_{\obs},
\ese
where $p_{\deb}$ and $p_{\obs}$ represent the proportion of patients whose WOMAC pain score improvement is greater than the MCID in the debridement group and the observation group, respectively.
Note that, one needs to obtain the evidence that the debridement group has a larger, instead of a merely different, proportion to claim the clinical importance of debriding the chondral lesions; therefore, the one-sided hypothesis testing is conducted.
It can be computed that the corresponding test statistic is $1.414$, p-value is $0.079$, and $95\%$ CI for difference in proportion is $[-0.053, 0.337]$.
Hence, we can conclude that there is no significant difference in the proportion of patients who achieve the clinical importance between deb and obs groups.
Again, the knowledge of treatment-specific MCID here would still not alter the conclusion that there is no need to debride the chondral lesions during the surgery.

\subsection{General Model Building}\label{sec:datageneral}

In addition, it is of interest to consider how to extend the treatment specific analysis to a more general model setting where multiple variables could be incorporated simultaneously.
To this end, we model the MCID with a linear structure and estimate the associated parameters following the method proposed in Section~\ref{sec:method}.
In these computations,
the tuning parameters $\delta$ and $\lambda$ are selected simultaneously from the pre-specified sets using 5-fold cross validation by minimizing the empirical version of the value function in (\ref{objfun3}).

We first fit Model 1 with only the treatment variable. Table~\ref{tab:app2} shows that its corresponding interval estimate covers zero hence there is no significant difference between the two treatment-specific MCIDs. This result is consistent with our previous result about the treatment specific clinical importance.
Then we fit Model 2 with $\Z$=\{treatment, age, gender, severity, BMI\}. Table~\ref{tab:app2} shows that, even after adjusting for the effects of age, gender, severity and BMI, there is still no significant difference between the two treatment-specific MCIDs.
Interestingly, in Model 2, the effects of age, gender, severity and BMI are all significant.
This result shows, a more senior patient, or a patient with a larger BMI, or a female patient (compared to male),  would generally need a greater threshold of pain improvement (hence harder to attain) to be claimed as clinically important.
Also,
the coefficient for severity is significantly smaller than zero. Its interpretation is that the more severely diseased patients would need a smaller threshold (hence easier to attain) to be claimed as clinically important.
This matches with our intuition. It is also consistent with prior work \citep{heald1997shoulder} which claimed that patients with more severe symptoms at baseline require less pain reduction to achieve clinical importance.


Finally, we conduct some sensitivity analyses on the values of $\lambda$ and $\delta$ in Model 2, shown in Figures~\ref{fig:app1} and \ref{fig:app2}.
Figure~\ref{fig:app1} demonstrates that, as $\lambda$ increases, all parameter estimates approach to zero. In particular, around a reasonable length of its optimal value, the parameter estimates do not vary much.
Similarly, Figure~\ref{fig:app2} also illustrates that the parameter estimates do not change drastically when the tuning parameter $\delta$ ranges from $0.1$ to $1$.

\section{Discussion}\label{sec:disc}

In this paper, we attempt to evaluate a randomized controlled trial from the perspective of assessing clinical importance, by proposing a novel principled statistical learning framework to estimate the magnitude of the minimal clinically important difference as well as to quantify its uncertainty. Our analysis result on the ChAMP trial reinforces the conclusion solely based on statistical significance, in the sense that the effect of debriding the chondral lesions during the knee surgery is not only statistically non-significant, but also clinically un-important. 


In terms of methodology, we concentrate on the investigation of MCID with a linear structure mainly for its ease of interpretation.
Technically, once the asymptotic representation for the unknown parameters is established,
that for the MCID can be straightforwardly derived. The MCID with a more flexible structure, such as semiparametric or even nonparametric, is certainly more ideal, but the corresponding technical derivation might be more challenging.

We also would like to mention that, although our motivation of this work as well as our case study is in orthopaedics with PROs derived from the WOMAC pain score and the SF-36 survey, the general framework proposed in this paper can be applied likewise to other disciplines with appropriate PRO instruments, such as lupus \citep{rai2015approaches}, cancer \citep{hui2018minimal}, Alzheimer's disease \citep{andrews2019disease}, among many others.

\begin{table}[htbp]
	\centering
	\caption{Simulation study (P1): the sample bias, sample standard deviation, estimated standard error, and coverage probability of 95\% confidence interval for the estimator of pMCID.
	}
	\label{tab:sim}
	\begin{tabular}{crrrrr}
		\hline\hline
		\multirow{2}{*}{Measure} & \multicolumn{5}{c}{Sample Size} \\
		& \multicolumn{1}{c}{600} & \multicolumn{1}{c}{900} & \multicolumn{1}{c}{1200} & \multicolumn{1}{c}{1500}  & \multicolumn{1}{c}{1800}  \\
		\hline
		Bias & 0.000 & 0.000 & 0.000 & 0.000 & 0.000 \\
		Standard Deviation (SD) & 0.008 & 0.007 & 0.006 & 0.005 & 0.005 \\
		Standard Error (SE) & 0.011 & 0.008 & 0.007 & 0.006 & 0.006 \\
		Coverage Probability (CP) & 0.950  & 0.962  & 0.966  & 0.972 & 0.956 \\
		\hline\hline
	\end{tabular}
\end{table}

\begin{figure}[htbp]
	\centering
	\includegraphics[scale=0.18]{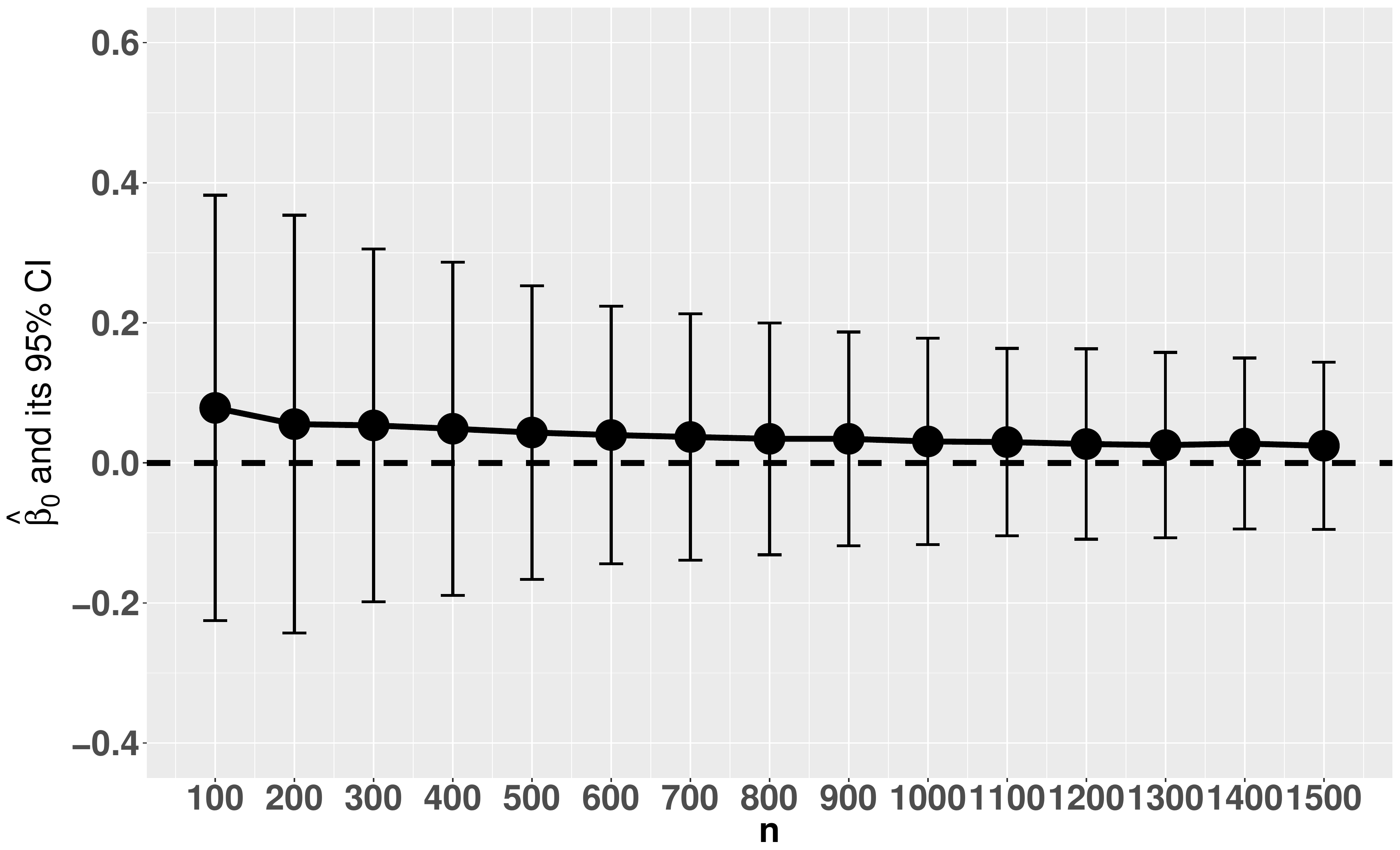}
	\includegraphics[scale=0.18]{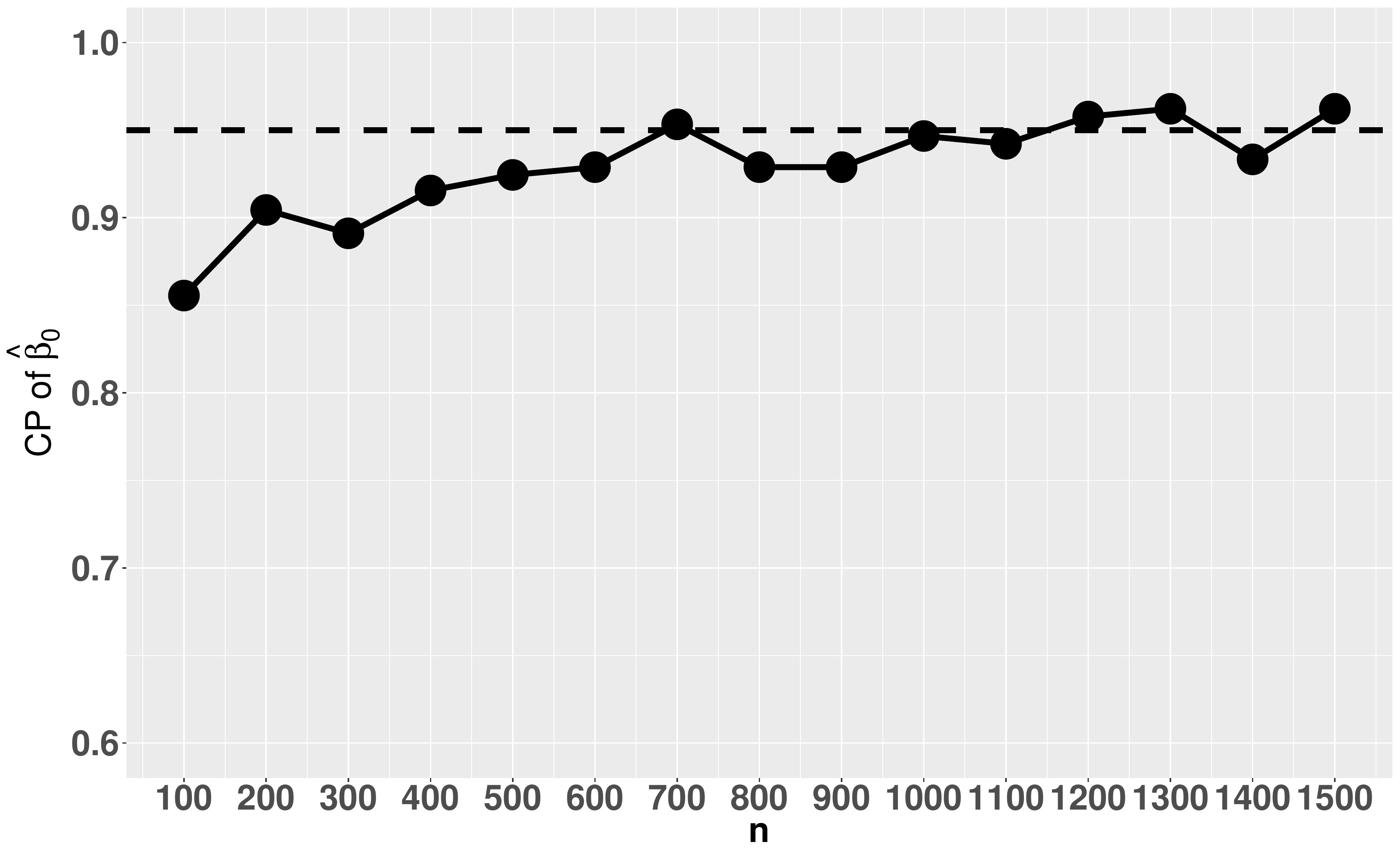}
	\includegraphics[scale=0.18]{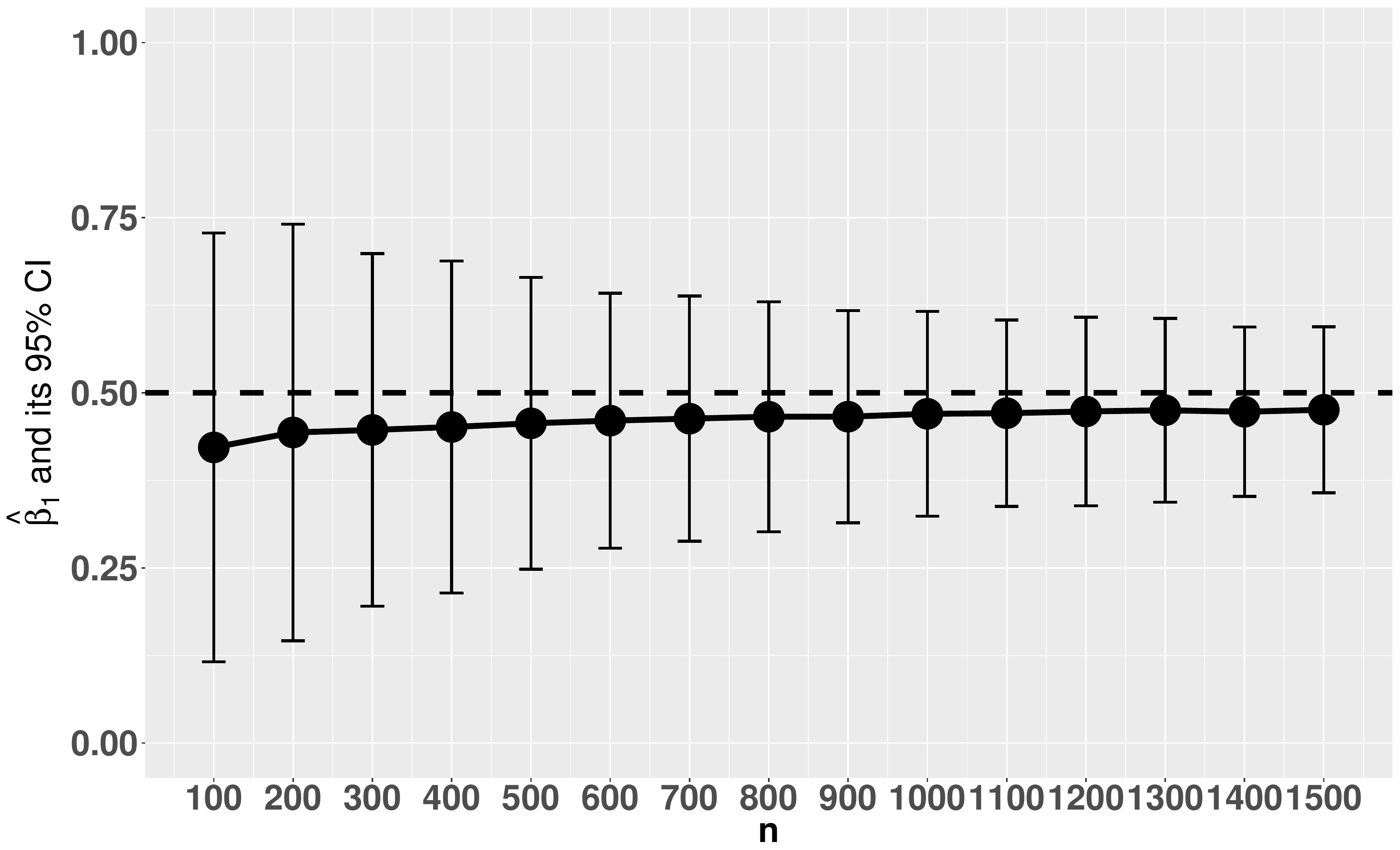}
	\includegraphics[scale=0.18]{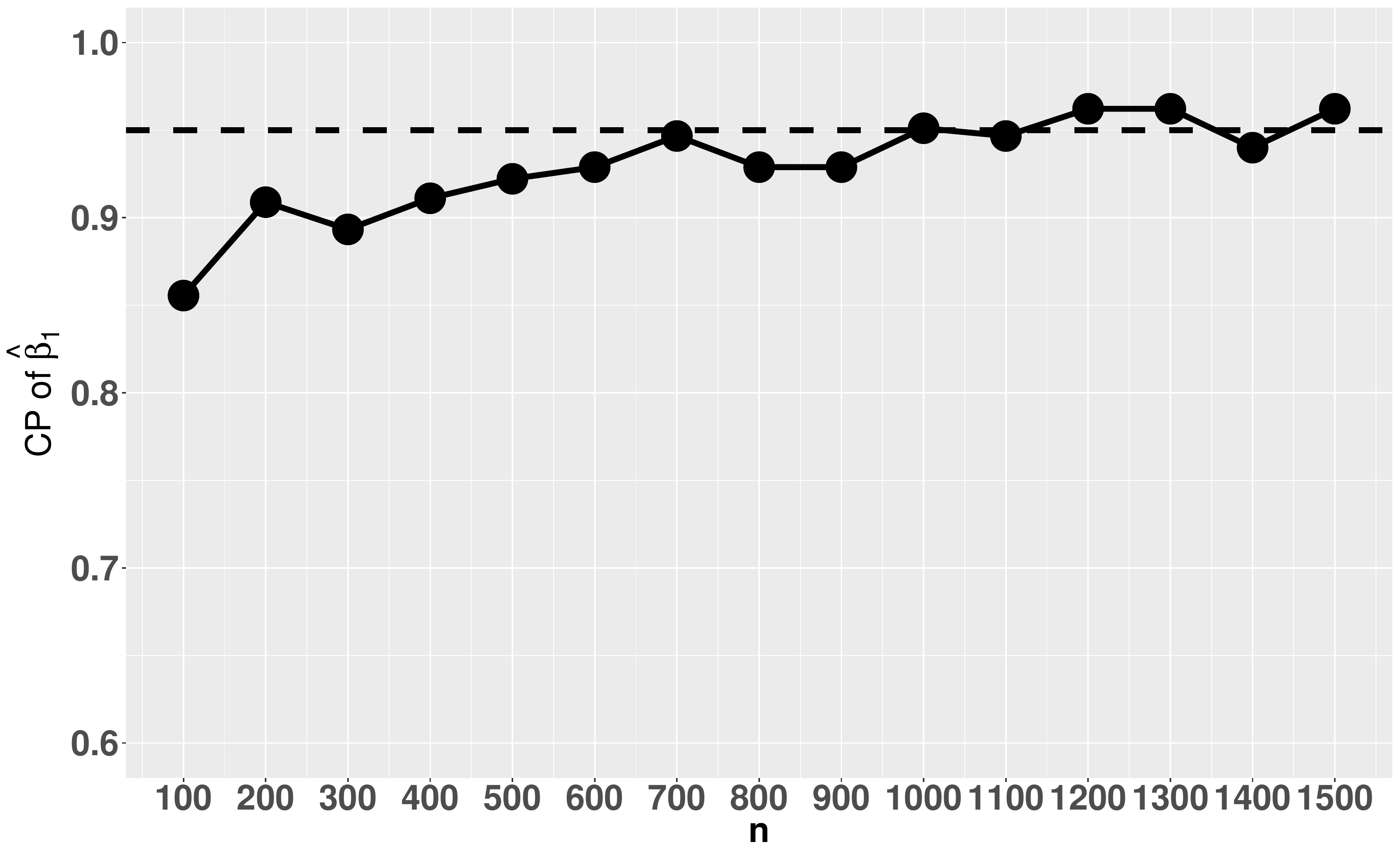}
	\caption{Simulation study (I1): the point estimator, 95\% confidence interval and its corresponding coverage probability for each of the parameters in iMCID.
	}
	\label{fig:sim1}
\end{figure}

\begin{figure}[htbp]
	\centering
	\includegraphics[scale=0.18]{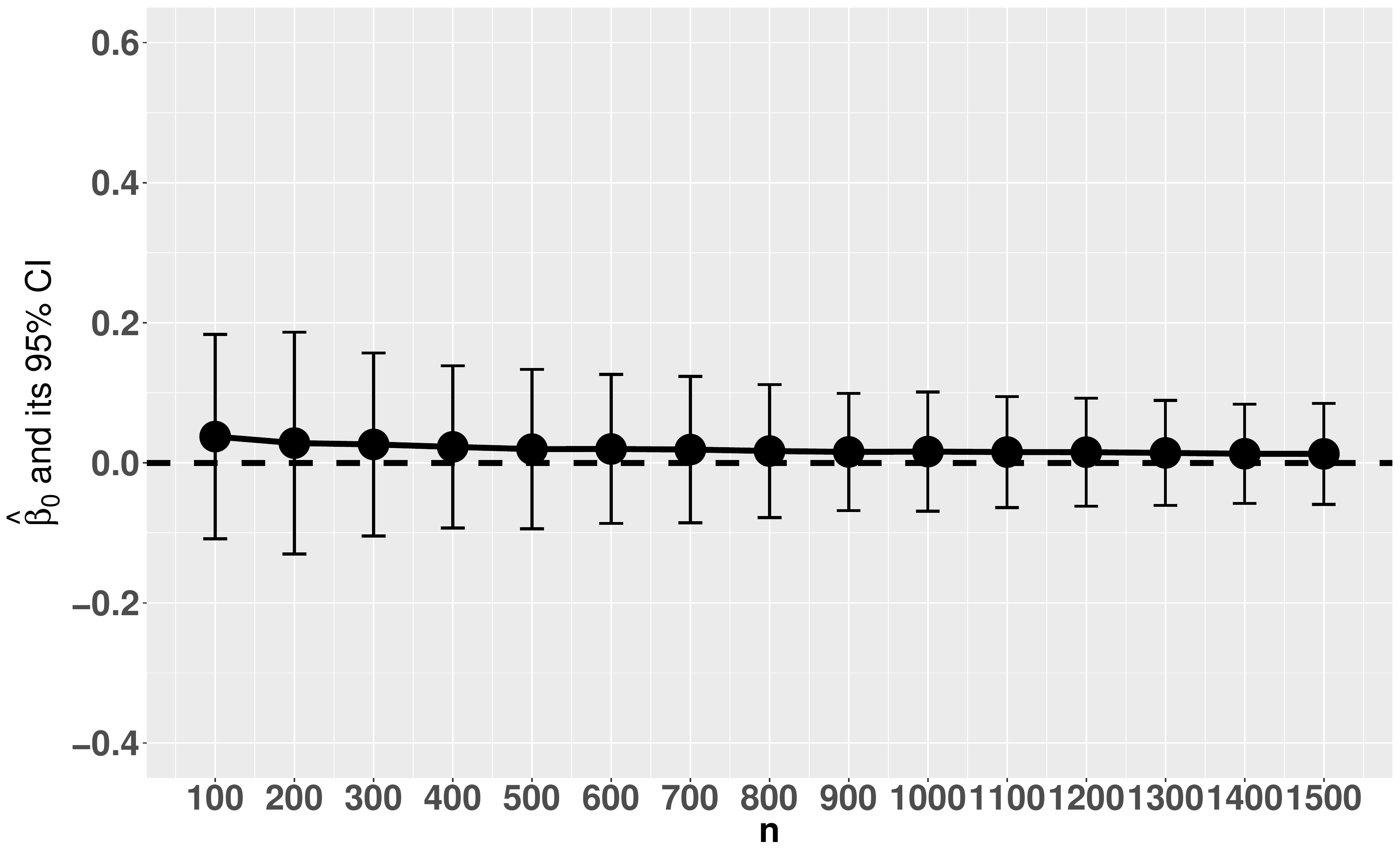}
	\includegraphics[scale=0.18]{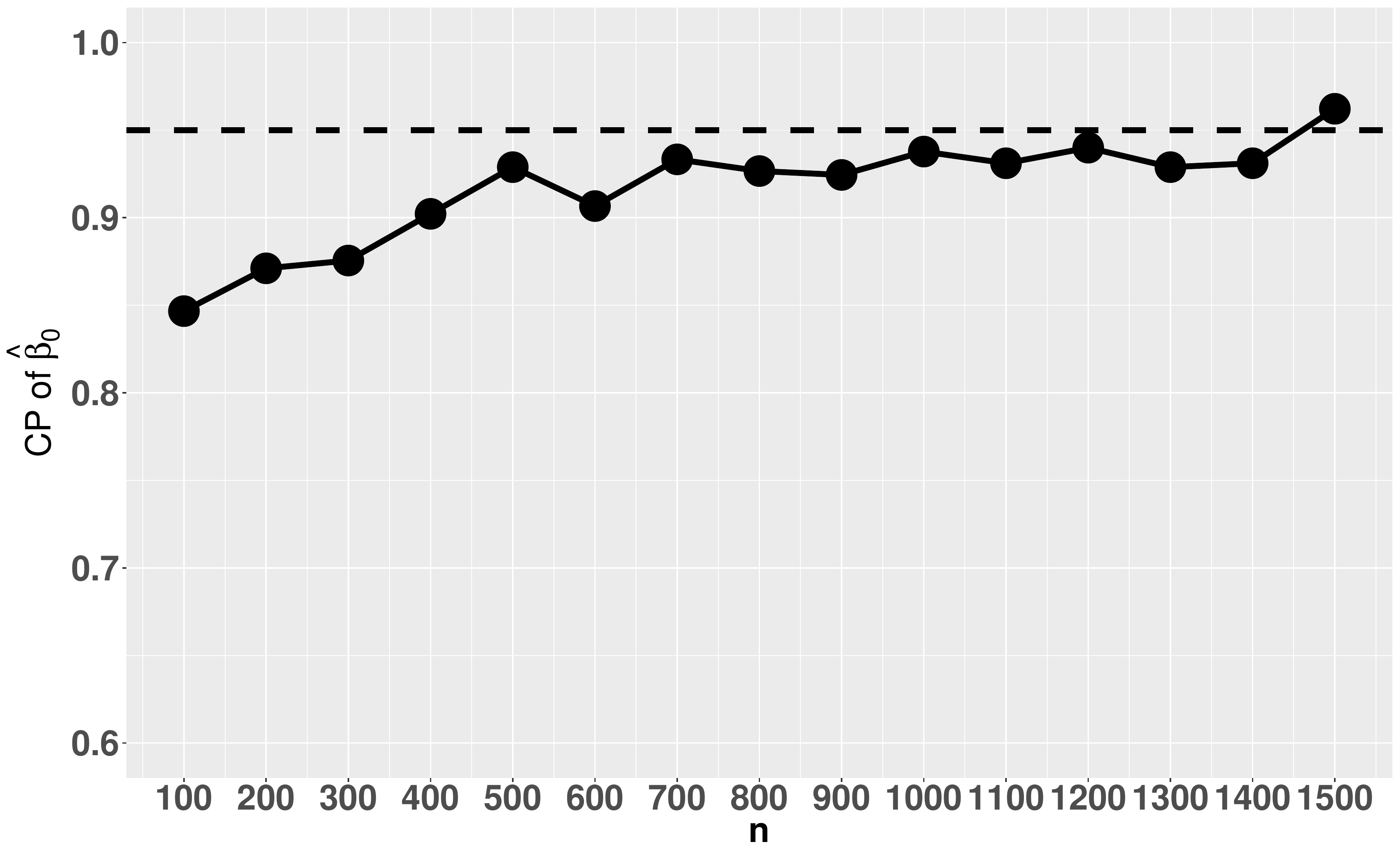}
	\includegraphics[scale=0.18]{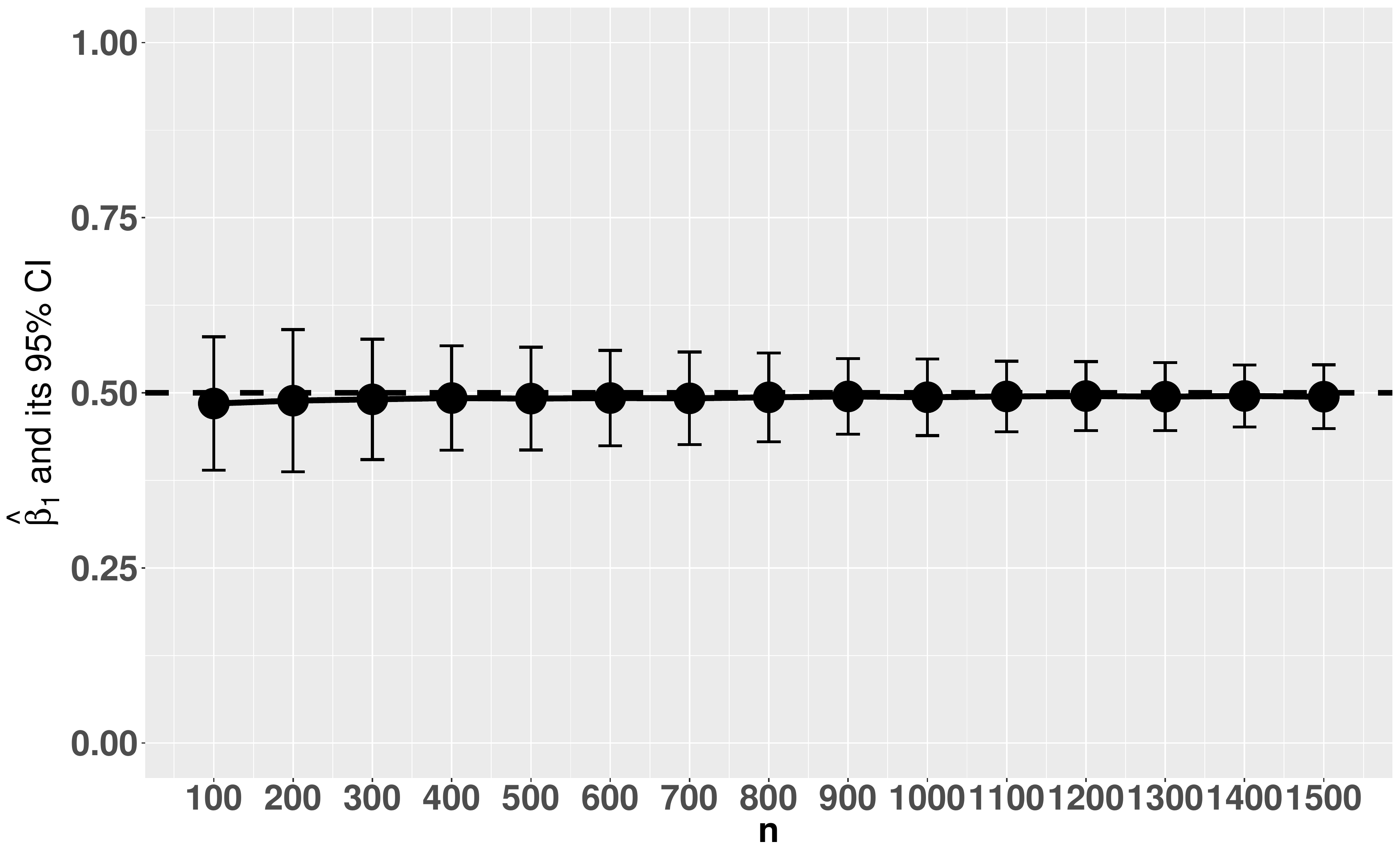}
	\includegraphics[scale=0.18]{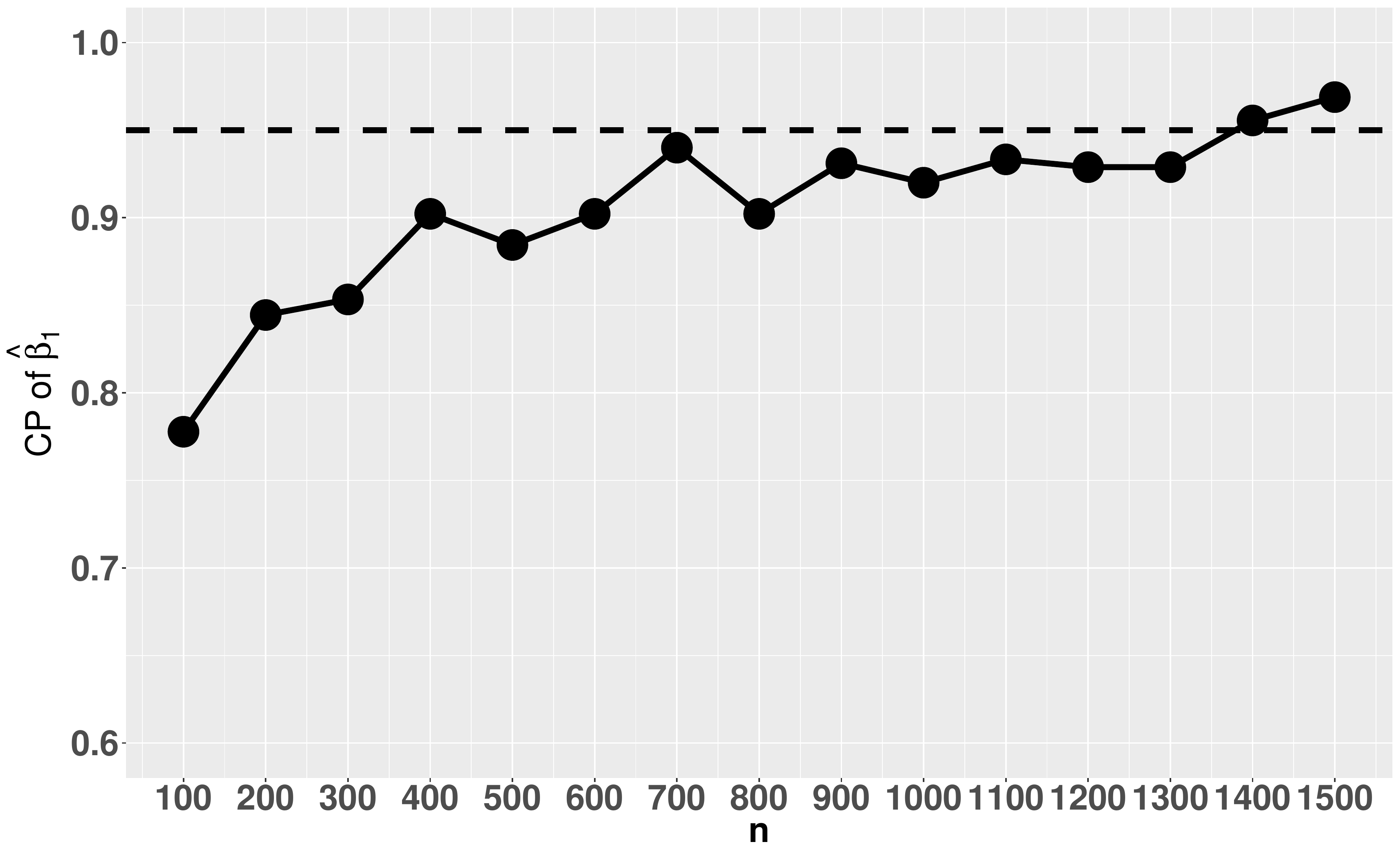}
	\includegraphics[scale=0.18]{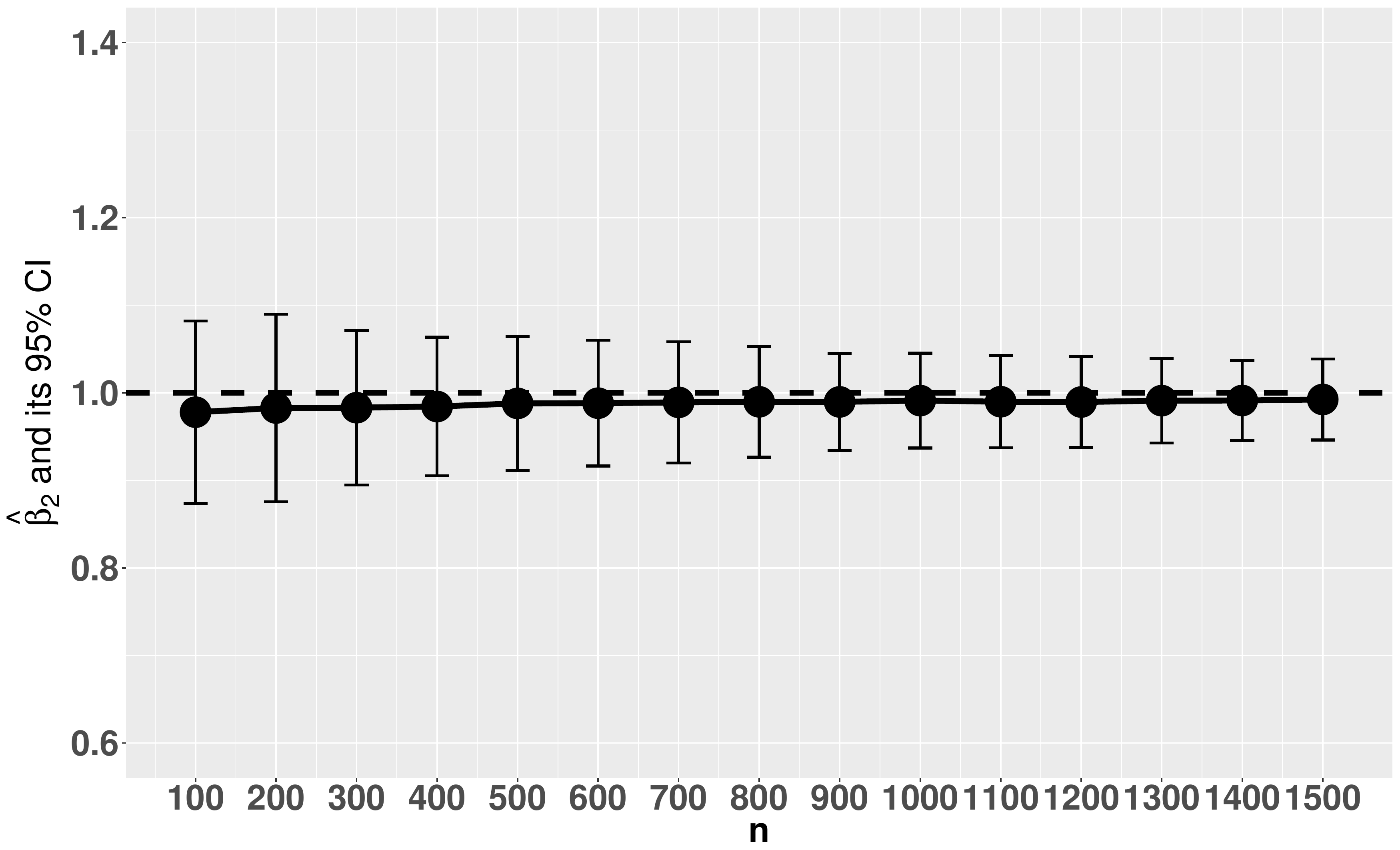}
	\includegraphics[scale=0.18]{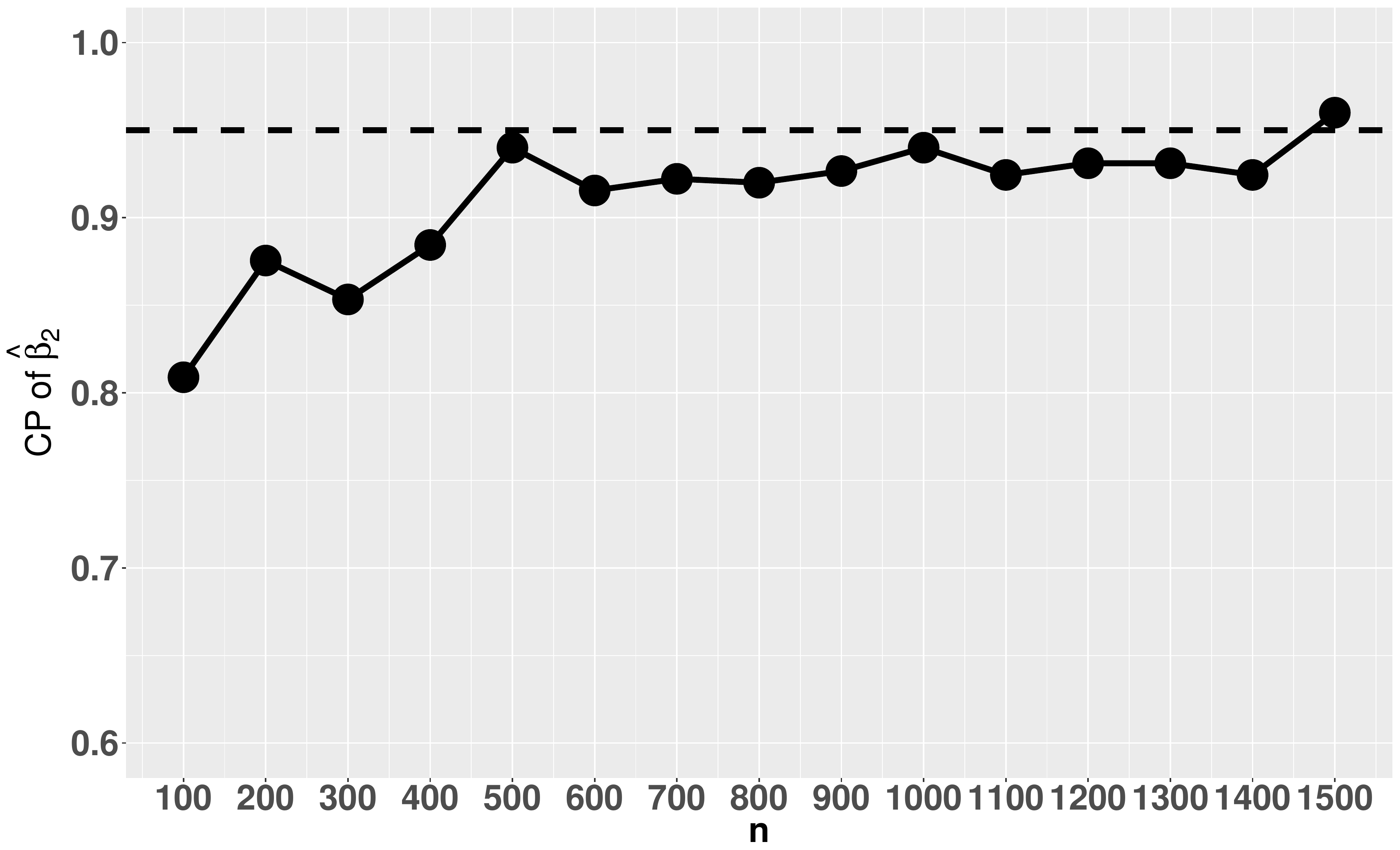}
	\caption{Simulation study (I2): the point estimator, 95\% confidence interval and its corresponding coverage probability for each of the parameters in iMCID.}
	\label{fig:sim2}
\end{figure}

\begin{figure}[htbp]
	\centering
	\includegraphics[scale=0.25]{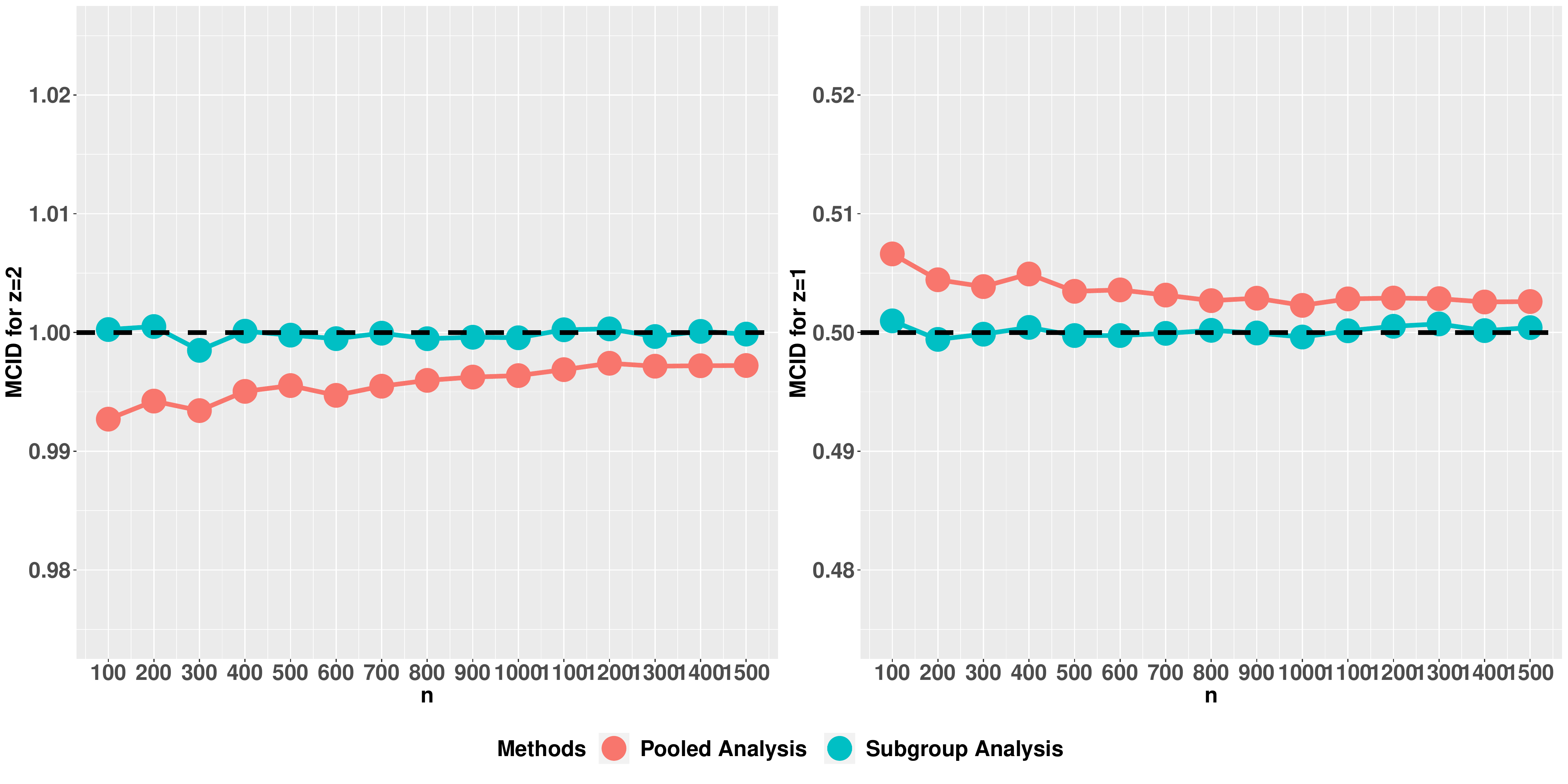}
	\caption{Simulation study (I3): the point estimator of the MCID using the subgroup analysis versus the pooled analysis.}
	\label{fig:sim31}
\end{figure}

\begin{table}[htbp]
	\centering
	\caption{In ChAMP trial application, the comparison results for the MCID, including the estimate, the standard error (SE) and the $95\%$ confidence interval (CI), for the whole population, the deb group (treatment=1), and the obs group (treatment=0), respectively.}
	\label{tab:app11}
	\begin{tabular}{ccrrrr}
		\hline\hline
		\multicolumn{1}{c}{\multirow{2}{*}{Method}} & \multicolumn{1}{c}{\multirow{2}{*}{Data Used}} & \multicolumn{1}{c}{\multirow{2}{*}{Estimate}} & \multicolumn{1}{c}{\multirow{2}{*}{SE}} & \multicolumn{2}{c}{$95\%$ CI} \\
		& & & & \multicolumn{1}{c}{Lower} & \multicolumn{1}{c}{Upper} \\
		\hline
		\multirow{3}{*}{Proposed}& All & 18.275 & 1.059 & 16.199 & 20.352 \\	
		& Treatment$=1$ & 17.945 & 1.485 & 15.035 & 20.856 \\
		& Treatment$=0$ & 18.598 & 8.979 & 1.000 & 36.195 \\
		\hline
		\multirow{3}{*}{Ad-hoc} & All & 29.718 & 2.031 & 25.668 & 33.768 \\
		& Treatment$=1$ & 28.438 & 3.348 & 21.608 & 35.267  \\
		& Treatment$=0$ & 30.769 & 2.502 & 25.704 & 35.834 \\
		\hline
		\multirow{3}{*}{ROC}& All & 22.500 & 10.013 & 7.500 & 52.500 \\
		& Treatment$=1$ & 7.500 & 18.516 & 2.500 & 57.500  \\
		& Treatment$=0$ & 22.500 & 6.686 & 12.500 & 42.500  \\
		\hline\hline
	\end{tabular}
\end{table}

\begin{table}[htbp]
	\centering
	\caption{In ChAMP trial application, the results for the parameters in iMCID with linear structure, including the estimate, the standard error (SE) and the $95\%$ confidence interval (CI). In Model 1, $Z$=treatment. In Model 2, $\Z$=\{treatment, age, gender, severity, BMI\}.}
	\label{tab:app2}
	\begin{tabular}{crrrr}
		\hline\hline
		\multicolumn{1}{c}{\multirow{2}{*}{Parameter}} & \multicolumn{1}{c}{\multirow{2}{*}{Estimate}} & \multicolumn{1}{c}{\multirow{2}{*}{SE}} & \multicolumn{2}{c}{$95\%$ CI} \\
		& & & \multicolumn{1}{c}{Lower} & \multicolumn{1}{c}{Upper} \\
		\hline
		\multicolumn{5}{l}{Model 1:} \\
		intercept & 14.097 & 7.220 & -0.054 & 28.248 \\
		treatment & -0.007 & 8.548 & -16.760 & 16.746 \\
		\hline
		\multicolumn{5}{l}{Model 2:} \\
		intercept & 0.131 & 5.615 & -10.875 & 11.137 \\
		treatment & 0.415 & 1.464 & -2.454 & 3.284 \\
		age & 0.308 & 0.041 & 0.229 & 0.388  \\
		gender & 3.142 & 1.319 & 0.557 & 5.726 \\
		severity & -2.236 & 0.396 & -3.013 & -1.460 \\
		BMI & 0.307 & 0.052 & 0.206 & 0.408 \\
		\hline\hline
	\end{tabular}
\end{table}

\begin{figure}[htbp]
	\centering
	\includegraphics[scale=0.55]{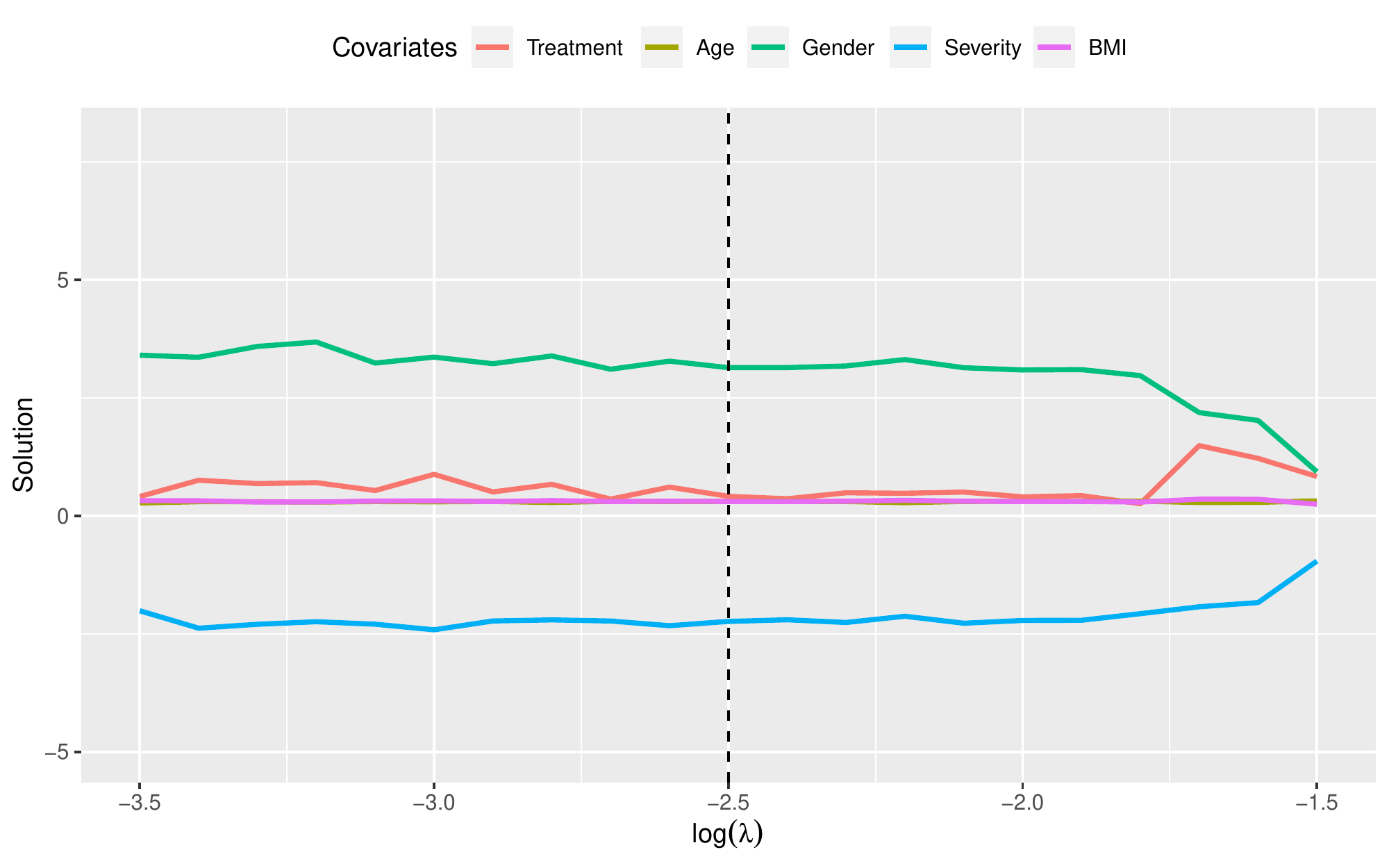}
	\caption{In ChAMP trial application, the solution path for $\wh\bb_\lambda$ as a function of $\lambda$ with $-3.5\leq\log(\lambda)\leq-1.5$. Dash line indicates the optimal $\lambda$ chosen by 5-fold cross validation.}
	\label{fig:app1}
\end{figure}

\begin{figure}[htbp]
	\centering
	\includegraphics[scale=0.55]{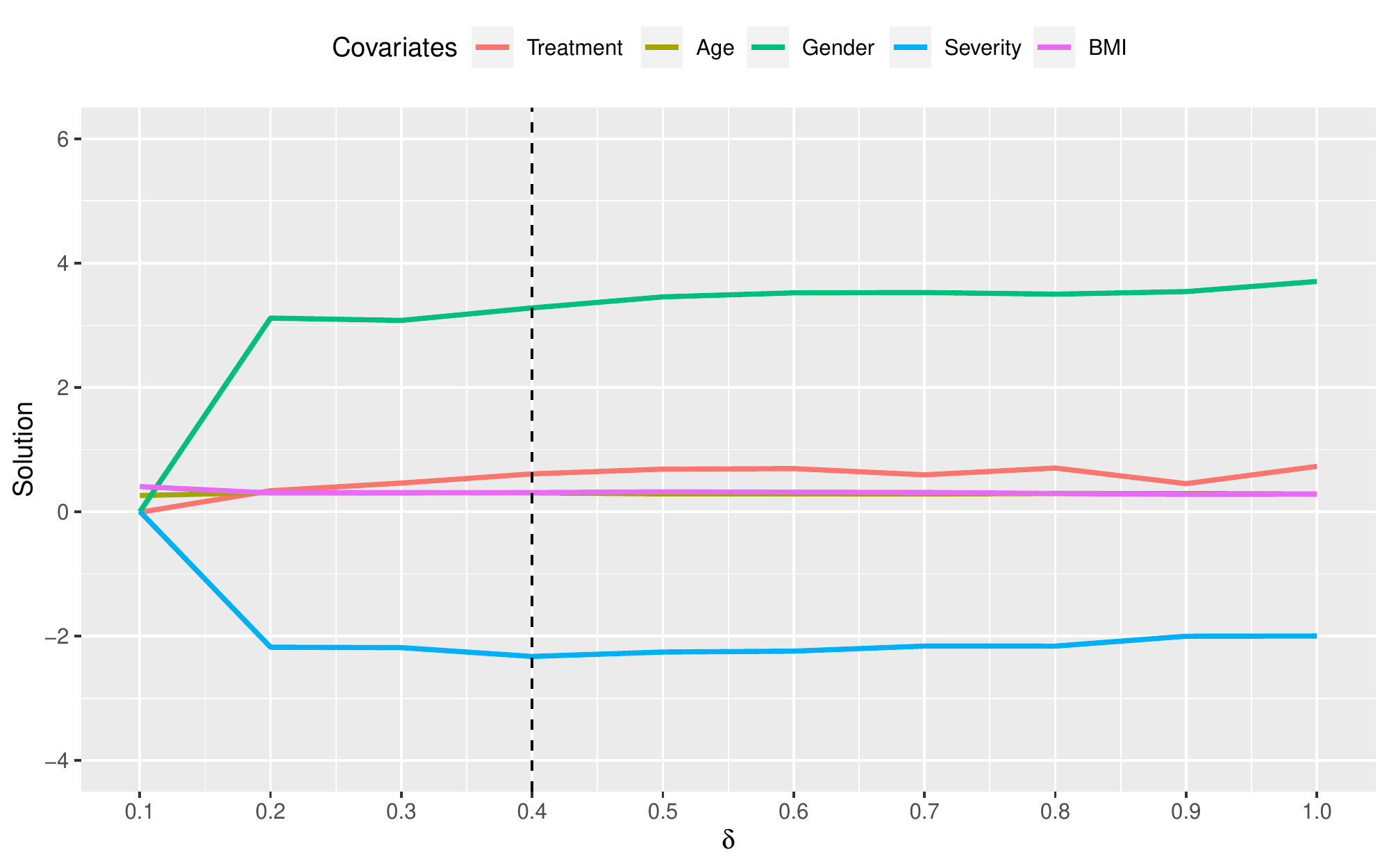}
	\caption{In ChAMP trial application, the solution path for $\wh\bb_\lambda$ as a function of $\delta$ with $0.1\leq\delta\leq1$. Dash line indicates the optimal $\delta$ chosen by 5-fold cross validation.}
	\label{fig:app2}
\end{figure}

\section*{Appendix}
\setcounter{equation}{0}\renewcommand{\theequation}{A.\arabic{equation}}
\setcounter{subsection}{0}\renewcommand{\thesubsection}{A.\arabic{subsection}}
\setcounter{table}{0}\renewcommand{\thetable}{A.\arabic{table}}

\begin{proof}[Proof of Lemma \ref{lemma1}]
	Starting from the objective function (\ref{objfun2}), we have
	\be
	& & \argmin_{\tau(\z)} ~ \frac12E\left[w(Y)\{1-Y\sign(X-\tau(\Z))\}|\Z=\z\right], \nonumber\\
	&=& \argmin_{\tau(\z)} ~ \frac12E_{Y|\Z}\big(E_{X|Y,\Z}\left[w(Y)\{1-Y\sign(X-\tau(\Z))\}|Y,\Z=\z\right]\big), \nonumber \\
	&=& \argmax_{\tau(\z)} ~ \left\{P(X\geq\tau(\Z)|Y=1,\Z=\z)-P(X\geq\tau(\Z)|Y=-1,\Z=\z)\right\}, \nonumber \\
	&=& \argmax_{\tau(\z)} ~ \left\{\int_{\tau(\z)}^{\infty}f(x-\a\trans\z)dx-\int_{\tau(\z)}^{\infty}f(x-\b\trans\z)dx\right\}.
	\label{lemma1_fun1}
	\ee
	Take the first derivative of the function in (\ref{lemma1_fun1}) and set it to 0, we have
	\bq
	-f(\tau(\z)-\a\trans\z)+f(\tau(\z)-\b\trans\z)=0.
	\label{lemma1_fun2}
	\eq
	Due to the symmetry and monotonicity of $f(t)$, we can solve  (\ref{lemma1_fun2}) by letting either $\tau(\z)-\a\trans\z=\tau(\z)-\b\trans\z$ or $\tau(\z)-a\trans\z+\tau(\z)-\b\trans(\z)=0$. Apparently the former equation cannot hold. Thus, the unique solution for equation (\ref{lemma1_fun2}) is
	\bsq
	\tau^0(\z)=(\bb^0)\trans\z=\left(\frac{\a+\b}2\right)\trans\z.
	\esq
	Next it is necessary to check whether $\tau^0(\z)$ is the maximizer. The second derivative of the function in (\ref{lemma1_fun1}) given $\tau(\z)=\tau^0(\z)$ is
	\bq -f^\prime(\tau^0(\z)-\a\trans\z)+f^\prime(\tau^0(\z)-\b\trans\z)=-f^\prime\left(\left(\frac{\b-\a}2\right)\trans\z\right)+f^\prime\left(\left(\frac{\a-\b}2\right)\trans\z\right).
	\label{lemma1_fun3}
	\eq
	We know that $(\b-\a)/2<0$ and $(\a-\b)/2>0$. According to the monotonicity of $f(t)$, if $\z\in\mathcal{R}^p_+$, it is easy to show $f^\prime((\b-\a)/2)\trans\z)\geq0$ and $f^\prime((\a-\b)/2)\trans\z)<0$; if $\z\not\in\mathcal{R}^p_+$, we let $\wt\z=\z-\min_i|z_i|=\z-\epsilon_z>0$, 
	so that $f^\prime((\b-\a)/2)\trans\wt\z)\geq0$ and $f^\prime((\a-\b)/2)\trans\wt\z)<0$. Therefore, the value of equation (\ref{lemma1_fun3}) is always below 0 and $\tau^0(\z)$ is the unique maximizer with an explicit form $\tau^0(\z)=(\bb^0)\trans\z$ where $\bb^0=(\a+\b)/2$.
\end{proof}

\begin{proof}[Proof of Lemma \ref{theorem1}]
	Let $\tau^*(\z)=(\bb^{*})\trans\z$. To show $\bb^*=\bb^0$, we only need to show $\tau^*(\z)=\tau^0(\z)$. For simplicity, we supress $\z$ in the proof so that $\tau^*(\z)=\tau^*$ and $\tau^0(\z)=\tau^0$, and denote $l_1(u)=l(u)$ in $L(u)$, $0<u\leq1/2$, and $l_2(u)=1-l(1-u)$ in $L(u)$, $1/2<u\leq1$. Let $g(\tau)$ denote the objective function in (\ref{objfun4}), $g^\prime(\tau)$ and $g^{\prime\prime}(\tau)$ its first and second order derivatives, then we have
	\allowdisplaybreaks
	\bse
	g^\prime(\tau)&=&E\left[w(Y)\left\{l_1^\prime\left(\frac{Y(X-\tau)}\delta\right)\ind_{\left\{0\leq \frac{Y(X-\tau)}\delta<1/2\right\}}+l_2^\prime\left(\frac{ Y(X-\tau)}\delta\right)\ind_{\left\{1/2\leq \frac{ Y(X-\tau)}\delta<1\right\}}\right\}\frac Y\delta\right], \\
	&=&
	\delta^{-1}\int_{\tau-\a\trans\z}^{\tau+\frac\delta2-\a\trans\z} l_1^\prime(\delta^{-1}(t+\a\trans\z-\tau))\{f(t+\a\trans\z-\b\trans\z-\delta)-f(t)\}dt \\
	& & +
	\delta^{-1}\int_{\tau-\a\trans\z}^{\tau+\frac\delta2-\a\trans\z} l_1^\prime(1/2-\delta^{-1}(t+\a\trans\z-\tau))\{f(t+\a\trans\z-\b\trans\z-\delta/2)-f(t+\delta/2)\}dt.
	\ese
	For $t\in[\tau-\a\trans\z,\tau-\a\trans\z+\delta/2]$, we have $\delta^{-1}(t+\a\trans\z-\tau)\in[0,1/2]$, $1/2-\delta^{-1}(t+\a\trans\z-\tau)\in[0,1/2]$, and $l_1(u)$ is decreasing. Thus, $l_1^\prime(\delta^{-1}(t+\a\trans\z-\tau))<0$ and $l_1^\prime(1/2-\delta^{-1}(t+\a\trans\z-\tau))<0$.
	
	If $\tau<\b\trans\z+\delta$, $f(t+\a\trans\z-\b\trans\z-\delta)-f(t)>0$ and  $f(t+\a\trans\z-\b\trans\z-\delta/2)-f(t+\delta/2)>0$ for $t\in[\tau-\a\trans\z,\tau-\a\trans\z+\delta/2]$ and the given condition of $\delta$, then $g^\prime(\tau)<0$; if $\tau>\a\trans\z-\delta$, $f(t+\a\trans\z-\b\trans\z-\delta)-f(t)<0$ and  $f(t+\a\trans\z-\b\trans\z-\delta/2)-f(t+\delta/2)<0$ for $t\in[\tau-\a\trans\z,\tau-\a\trans\z+\delta/2]$ and the given condition of $\delta$, then $g^\prime(\tau)>0$.
	Hence, the solution of $g^\prime(\tau)>0$ cannot exist outside of $[\b\trans\z+\delta, \a\trans\z-\delta]$. To prove the unique solution exists between $\b\trans\z+\delta$ and $\a\trans\z-\delta$, we need to show
	that $g^\prime(\tau)$ is increasing for $\tau\in[\b\trans\z+\delta,\a\trans\z-\delta]$, that is, $g^{\prime\prime}(\tau)>0$ for $\tau\in[\b\trans\z+\delta,\a\trans\z-\delta]$.
	One can verify
	\bse
	g^{\prime\prime}(\tau)&=&\delta^{-2}E\left\{w(Y)\left(l_1^{\prime\prime}\left(\frac{Y(X-\tau)}\delta\right)\ind_{\left\{0\leq \frac{Y(X-\tau)}\delta<\frac12\right\}}+l_2^{\prime\prime}\left(\frac{Y(X-\tau)}\delta\right)\ind_{\left\{1/2\leq \frac{Y(X-\tau)}\delta<1\right\}}\right)\right\}, \\ 
	&=&\delta^{-2}\int_{\tau-\a\trans\z}^{\tau-\a\trans\z+\frac\delta2} l_1^{\prime\prime}(\delta^{-1}(t+\a\trans\z-\tau))\{f(t)-f(t+\frac\delta2)\}dt \\
	& & + \delta^{-2}\int_{\tau-\b\trans\z-\frac\delta2}^{\tau-\b\trans\z} l_1^{\prime\prime}(-\delta^{-1}(t+\b\trans\z-\tau))\{f(t)-f(t-\frac\delta2)\}dt.
	\ese
	Since $l_1(u)$ is concave, $l_1^{\prime\prime}(u)<0$. If $\tau\in[\b\trans\z+\delta,\a\trans\z-\delta]$, then $f(t)<f(t+\frac\delta2)$ for $t\in[\tau-\a\trans\z, \tau-\a\trans\z+\frac\delta2]$ and $f(t)<f(t-\frac\delta2)$ for $t\in[\tau-\b\trans\z-\frac\delta2, \tau-\b\trans\z]$ and the given condition of $\delta$. Therefore, $g^{\prime\prime}(\tau)>0$ and we can prove that $g^\prime(\tau)=0$ has a unique solution.
	
	Since $\tau^0=(\bb^0)\trans\z=\{(\a+\b)/2\}\trans\z$, it is easy to show that $g^\prime(\tau^0)=0$.
	Now, we have proved that the solution of (\ref{objfun3}) is the same as that of (\ref{objfun4}). Since $\tau$ is a linear function of $\z$, for any given $\z$, we have $\bb^*=\bb^0$. This completes the proof.
\end{proof}

\begin{proof}[Proof of Theorem~\ref{theorem2}]
	This proof needs a preliminary result on the property of $L(u)$. We first present it as a Lemma and then prove it.
	\begin{lemma}
		For the surrogate loss family, decompose $L(u)$ as $L(u)=L_1(u)-L_2(u)$, where $L_1(u)=\{-k(u-1/2)+1/2\}\ind_{\{u\leq1/2\}}+\{1-l(1-u)\ind_{\{1/2<u\leq1\}}\}$ and $L_2(u)=\{-k(u-1/2)-1/2\}\ind_{\{u\leq0\}}+\{-k(u-1/2)+1/2-l(u)\}\ind_{\{0<u\leq1/2\}}$.
		Assume $R_1=L_1(u)-L_1(v)-L_1^\prime(v)(u-v)$ and $R_2=L_2(u)-L_2(v)-L_2^\prime(v)(u-v)$,
		and let $s_1, s_2, s_3, s_4$ be the sets where $s_1=\{(v,u): v\leq\frac12,u>\frac12 ~ \text{or} ~ \frac12<v\leq1,u\leq \frac12 ~ \text{or} ~ v>1,u\leq1 ~ \text{or} ~ \frac12<v\leq1,u>1\}$, $s_2=\{(v,u): \frac12<v\leq1,\frac12<u\leq1\}$, $s_3=\{(v,u): v\leq 0,u>0 ~ \text{or} ~ 0<v\leq \frac12,u\leq0 ~ \text{or} ~ v>\frac12,u\leq \frac12 ~ \text{or} ~ 0<v\leq \frac12,u>\frac12\}$, and $s_4=\{(v,u): 0<v\leq \frac12,0<u\leq \frac12\}$.
		Then, $s_1$ and $s_2$ are disjoint, $s_3$ and $s_4$ are disjoint, $0\leq R_1\leq k|u-v|\ind_{\{(v,u)\in s_1\}}+M|u-v|\ind_{\{(v,u)\in s_2\}}$, and
		$0\leq R_2\leq k|u-v|\ind_{\{(v,u)\in s_3\}}+M|u-v|\ind_{\{(v,u)\in s_4\}}$.
		\label{lemma2}
	\end{lemma}
	\begin{proof}[Proof of Lemma \ref{lemma2}]
		For simplicity, let $l_1(u)=l(u)$ in $L(u)$, $0<u\leq1/2$, and $l_2(u)=1-l(1-u)$ in $L(u)$, $1/2<u\leq1$. According to the definition of $R_1$, we have
		\bse
		R_1&=&\{k(\frac12-u)+\frac12\}\ind_{\{u\leq \frac12\}}+l_2(u)\ind_{\{\frac12<u\leq 1\}} -\{k(\frac12-v)+\frac12\}\ind_{\{v\leq \frac12\}}-l_2(v)\ind_{\{\frac12<v\leq 1\}} \\
		&&+\left(k\ind_{\{v\leq \frac12\}}-l_2^\prime(v)\ind_{\{\frac12<v\leq1\}}\right)(u-v).
		\ese
		First for $v\leq\frac12$,
		\allowdisplaybreaks
		\bse
		R_1&=&\{k(\frac12-u)+\frac12\}\ind_{\{u\leq \frac12\}}+l_2(u)\ind_{\{\frac12<u\leq 1\}}-\{k(\frac12-v)+\frac12\}+k(u-v) \\
		&=&\{l_2(u)-k(\frac12-u)-\frac12\}\ind_{\{\frac12<u\leq1\}}+\{k(u-\frac12)-\frac12\}\ind_{\{u>1\}}.
		\ese
		If $\frac12<u\leq1$, then $R_1=l_2(u)-k(\frac12-u)-\frac12$. Due to the convexity of $l_2$, we have
		\bse
		R_1&\geq& l_2(\frac12)+l_2^\prime(\frac12)(u-\frac12)-k(\frac12-u)-\frac12=0, \\
		R_1&\leq&-k(\frac12-u)\leq-k(v-u)=k(u-v), \quad \text{by $l_2(u)\leq \frac12$ and $v\leq \frac12$.}
		\ese
		If $u>1$, then $R_1=k(u-\frac12)-\frac12$. Since $-\frac12k-\frac12\geq -k$ derived from $l_2(1)\geq l_2(\frac12)+l_2^\prime(\frac12)(1-\frac12)$, we can get
		$R_1\geq k(u-1)\geq0$ and $R_1\leq k(u-\frac12)\leq k(u-v)$.
		Hence,
		\bse
		R_1\leq k(u-v)\ind_{\{v\leq \frac12,\frac12<u\leq1\}}+k(u-v)\ind_{\{v\leq \frac12,u>1\}}
		= k(u-v)\ind_{\{v\leq \frac12,u>\frac12\}}.
		\ese
		Then for $v>1$, $R_1=\{k(\frac12-u)+\frac12\}\ind_{\{u\leq \frac12\}}+l_2(u)\ind_{\{\frac12<u\leq 1\}}$.
		If $u\leq \frac12$, then $R_1=k(\frac12-u)+\frac12$, so $R_1\geq k(\frac12-u)\geq0$ and $R_1\leq k(1-u)\leq k(v-u)$.
		If $\frac12<u\leq1$, then $R_1=l_2(u)$. By the convexity of $l_2$ and $v>1$, we have $R_1\geq 0$ and $R_1\leq l_2(1)-l_2^\prime(u)(1-u)=-l_2^\prime(u)(1-u)\leq -l_2^\prime(\frac12)(1-u)\leq k(v-u)$.
		Thus, \bse
		R_1 \leq k(v-u)\ind_{\{v>1,u\leq \frac12\}}+k(v-u)\ind_{\{v>1,\frac12<u\leq1\}}
		= k(v-u)\ind_{\{v>1,u\leq1\}}.
		\ese
		At last for $\frac12<v\leq1$,
		\bse
		R_1&=&\{k(\frac12-u)+\frac12\}\ind_{\{u\leq \frac12\}}+l_2(u)\ind_{\{\frac12<u\leq 1\}}-l_2(v)-l_2^\prime(v)(u-v) \\
		&=&\{k(\frac12-u)+\frac12-l_2(v)-l_2^\prime(v)(u-v)\}\ind_{\{u\leq \frac12\}} \\
		& &+\{l_2(u)-l_2(v)-l_2^\prime(v)(u-v)\}\ind_{\{\frac12<u\leq 1\}} +\{-l_2(v)-l_2^\prime(v)(u-v)\}\ind_{\{u>1\}}.
		\ese
		If $u\leq \frac12$, then $R_1=k(\frac12-u)+\frac12-l_2(v)-l_2^\prime(v)(u-v)$. Let $v^*$ be a fixed value between $u$ and $v$, then
		$R_1=k(\frac12-u)+l_2(\frac12)-l_2(v)-l_2^\prime(v)(\frac12-v)+l_2^\prime(v)(\frac12-u) \geq 0$, and $R_1=k(\frac12-u)+\frac12-l_2(v)-l_2^\prime(v)(u-v)\leq k(v-u)$.
		If $\frac12<u\leq1$, then $R_1=l_2(u)-l_2(v)-l_2^\prime(v)(u-v)$. Obviously by the convexity of $l_2$, we have $R_1\geq 0$ and $R_1=l_2(u)-l_2(v)-l_2^\prime(v)(u-v)\leq M(u-v)^2$.
		If $u>1$, then $R_1=-l_2(v)-l_2^\prime(v)(u-v)=l_2(1)-l_2(v)-l_2^\prime(v)(u-v)$, therefore,
		\bse
		R_1&\geq& l_2(1)-l_2(v)-l_2^\prime(v)(1-v)\geq0 \quad\text{by the convexity of $l_2$, $l_2(1)=0$ and $u>1$}, \\
		R_1&\leq&-l_2^\prime(v)(u-v)\leq k(u-v) \quad\text{by $l_2(v)\geq0$ and $l_2^\prime(v)\geq -k$}.
		\ese
		Therefore, $R_1 \leq k(v-u)\ind_{\{\frac12<v\leq1,u\leq \frac12\}}+M(u-v)^2\ind_{\{\frac12<v\leq1,\frac12<u\leq1\}}+k(u-v)\ind_{\{\frac12<v\leq1,u>1\}}$.
		In summary, $R_1$ is always nonnegative, and
		\bse
		R_1
		&\leq&k|u-v|\ind_{\{v\leq \frac12,u>\frac12 ~ \text{or} ~ \frac12<v\leq1,u\leq \frac12 ~ \text{or} ~ v>1,u\leq1 ~ \text{or} ~ \frac12<v\leq1,u>1\}}+M|u-v|^2\ind_{\{\frac12<v\leq1,\frac12<u\leq1\}} \\
		&=&k|u-v|\ind_{\{(v,u)\in s_1\}}+M|u-v|\ind_{\{(v,u)\in s_2\}}.
		\ese
		Next, one can similarly prove that $R_2$ is also always nonnegative, and
		\bse
		R_2
		&\leq&k|u-v|\ind_{\{v\leq 0,u>0 ~ \text{or} ~ 0<v\leq \frac12,u\leq 0 ~ \text{or} ~ v>\frac12,u\leq \frac12 ~ \text{or} ~ 0<v\leq \frac12,u>\frac12\}}+M|u-v|^2\ind_{\{0<v\leq \frac12,0<u\leq \frac12\}}, \\
		&=&k|u-v|\ind_{\{(v,u)\in s_3\}}+M|u-v|\ind_{\{(v,u)\in s_4\}}.
		\ese
	\end{proof}
	Now, for fixed $\tth\in\mathcal{R}^p$, define $\Lambda_n(\tth)=n(\ml_{\lambda}(\bb^*+\tth/\sqrt{n})-\ml_{\lambda}(\bb^*))$ and $\Gamma_n(\tth)=E\Lambda_n(\tth)$.
	Observe that $\Gamma_n(\tth)=n\left(\ml(\bb^*+\tth/\sqrt{n})-\ml(\bb^*)\right)+\frac\lambda2\left(||\tth_{-1}||^2+2\sqrt{n}(\bb^{*})\trans\tth_{-1}\right)$,
	where $\tth_{-1}=\tth\setminus\theta_1$ and $\theta_1$ is the first item in $\tth$. Expanding $\ml(\bb^*+\tth/\sqrt{n})$ around $\bb^*$ through Taylor expansion, we have
	\bse
	\Gamma_n(\tth)
	&= & \frac12\tth\trans \H(\wb\bb)\tth+\frac\lambda2\left(||\tth_{-1}||^2+2\sqrt{n}(\bb^*)\trans\tth_{-1}\right),
	\ese
	where $\wb\bb=\bb^*+(t/\sqrt{n})\tth$ for some $0<t<1$.
	Define $D_{jk}(\rr)=\H(\bb^*+\rr)_{jk}-\H(\bb^*)_{jk}$ for $0\leq j\leq p,0\leq k\leq p$. Since $\H(\wb\bb)$ is continuous in $\wb\bb$, there exists $\delta_0>0$ such that $|D_{jk}(\rr)|<\epsilon_0$ if $||\rr||<\delta_0$ for any $\epsilon_0>0$ and all $0\leq j\leq p,0\leq k\leq p$. Then as $n\rightarrow\infty$, since $||(t/\sqrt{n})\tth||<\delta_0$, we have
	\bse
	|\tth\trans(\H(\wb\bb)-\H(\bb^*))\tth| \leq \sum_{j,k}|\theta_j||\theta_k||D_{j,k}(t/\sqrt{n}\tth)| \leq \epsilon_0\sum_{j,k}|\theta_j||\theta_k| \leq 2\epsilon_0||\tth||^2,
	\ese
	and hence $\frac12\tth\trans \H(\wb\bb)\tth=\frac12\tth\trans \H(\bb^*)\tth+o(1)$.
	Together with the assumption that $\lambda=o(n^{-1/2})$, we have $\Gamma_n(\tth)=\frac12\tth\trans\H(\bb^*)\tth+o(1)$.
	For simplicity, denote $l_1(u)=l(u)$ in $L(u)$, $0<u\leq1/2$, and $l_2(u)=1-l(1-u)$ in $L(u)$, $1/2<u\leq1$. Then define
	\be
	\W_n&=& \sumi\left\{w_iL_{\delta1}^\prime\left(\xi_i(\bb^*)\right)\xxi_i^\prime(\bb^*)\right\}-\sumi\left\{w_iL_{\delta2}^\prime\left(\xi_i(\bb^*)\right)\xxi_i^\prime(\bb^*)\right\} \nonumber \\ &=&-\delta^{-1}\sumi\left[\left\{l_1^\prime\left(\xi_i(\bb^*)/\delta\right)\ind_{\left\{0<\xi_i(\bb^*)\leq \frac\delta2\right\}}+l_2^\prime\left(\xi_i(\bb^*)/\delta\right)\ind_{\left\{\frac\delta2<\xi_i(\bb^*)\leq\delta\right\}}\right\}w_iy_i\z_i\right],
	\ee
	where $\W_n$ is a $p-$dimensional vector. Then $\frac{1}{\sqrt{n}}\W_n$ follows asymptotically $N(0,n\G(\bb^*))$ by central limit theorem, where
	\bsq \G(\bb^*)=\delta^{-2}E\left[\left\{\left(l_1^{\prime}\left(\xi(\bb^*)/\delta\right)\right)^2\ind_{\left\{0<\xi(\bb^*)\leq \frac\delta2\right\}}+\left(l_2^{\prime}\left(\xi(\bb^*)/\delta\right)\right)^2\ind_{\left\{\frac\delta2<\xi(\bb^*)\leq \delta\right\}}\right\}w(Y)^2\Z\Z\trans\right].
	\esq
	This is because $E\left\{w_iL_{\delta1}^\prime\left(\xi_i(\bb^*)\right)\xxi_i^\prime(\bb^*)-w_iL_{\delta2}^\prime\left(\xi_i(\bb^*)\right)\xxi_i^\prime(\bb^*)\right\}=0$ by $S(\bb^*)=0$, and
	\bse
	& & E\left\{w_iL_{\delta1}^\prime\left(\xi_i(\bb^*)\right)\xxi_i^\prime(\bb^*)-w_iL_{\delta2}^\prime\left(\xi_i(\bb^*)\right)\xxi_i^\prime(\bb^*)\right\}^2 \nonumber \\ &=&\delta^{-2}E\left[\left\{(l_1^{\prime}\left(\xi_i(\bb^*)/\delta\right))^2\ind_{\left\{0<\xi_i(\bb^*)\leq \frac\delta2\right\}}+(l_2^{\prime}\left(\xi_i(\bb^*)/\delta\right))^2\ind_{\left\{\frac\delta2<\xi_i(\bb^*)\leq \delta\right\}}\right\}w_i^2\z_i\z_i\trans\right].
	\ese
	Now define $R_{i,n}(\tth)=R_{1i,n}(\tth)-R_{2i,n}(\tth)$, where
	\bse R_{1i,n}(\tth)&=&w_iL_{\delta1}\left(\xi_i(\bb^*+\tth/\sqrt{n})\right)-w_iL_{\delta1}\left(\xi_i(\bb^*)\right)-w_iL_{\delta1}^\prime\left(\xi_i(\bb^*)\right)\left(\xi_i(\bb^*+\tth/\sqrt{n})-\xi_i(\bb^*)\right), \\ R_{2i,n}(\tth)&=&w_iL_{\delta2}\left(\xi_i(\bb^*+\tth/\sqrt{n})\right)-w_iL_{\delta2}\left(\xi_i(\bb^*)\right)-w_iL_{\delta2}^\prime\left(\xi_i(\bb^*)\right)\left(\xi_i(\bb^*+\tth/\sqrt{n})-\xi_i(\bb^*)\right),
	\ese
	and we have $\Lambda_n(\tth)=\Gamma_n(\tth)+\W_n\trans\tth/\sqrt{n}+\sumi(R_{i,n}(\tth)-ER_{i,n}(\tth))$, and
	\be
	R_{1i,n}(\tth)&\leq& k|\delta^{-1}\z_i\trans\tth/\sqrt{n}|\ind_{\{(\xi_i(\bb^*),\xi_i(\bb^*+\tth/\sqrt{n}))\in s_{1i}\}}+M|\delta^{-1}\z_i\trans\tth/\sqrt{n}|^2\ind_{\{(\xi_i(\bb^*),\xi_i(\bb^*+\tth/\sqrt{n}))\in s_{2i}\}}, \label{ineR1} \\
	R_{2i,n}(\tth)&\leq& k|\delta^{-1}\z_i\trans\tth/\sqrt{n}|\ind_{\{(\xi_i(\bb^*),\xi_i(\bb^*+\tth/\sqrt{n}))\in s_{3i}\}}+M|\delta^{-1}\z_i\trans\tth/\sqrt{n}|^2\ind_{\{(\xi_i(\bb^*),\xi_i(\bb^*+\tth/\sqrt{n}))\in s_{4i}\}}, \label{ineR2}
	\ee
	where
	\bse
	s_{1i}&=&\{(\xi_i(\bb^*),\xi_i(\bb^*+\tth/\sqrt{n})): \xi_i(\bb^*)\leq \frac\delta2,\xi_i(\bb^*+\tth/\sqrt{n})>\frac\delta2 ~ \text{or} ~ \frac\delta2<\xi_i(\bb^*)\leq\delta, \\
	&& \xi_i(\bb^*+\tth/\sqrt{n})\leq \frac\delta2 ~ \text{or} ~ \xi_i(\bb^*)>\delta,\xi_i(\bb^*+\tth/\sqrt{n})\leq\delta ~ \text{or} ~ \frac\delta2<\xi_i(\bb^*)\leq\delta,\xi_i(\bb^*+\tth/\sqrt{n})>\delta\}, \\
	s_{2i}&=&\{(\xi_i(\bb^*),\xi_i(\bb^*+\tth/\sqrt{n})): \frac\delta2<\xi_i(\bb^*)\leq\delta,\frac\delta2<\xi_i(\bb^*+\tth/\sqrt{n})\leq\delta\}, \\
	s_{3i}&=&\{(\xi_i(\bb^*),\xi_i(\bb^*+\tth/\sqrt{n})): \xi_i(\bb^*)\leq 0,\xi_i(\bb^*+\tth/\sqrt{n})>0 ~ \text{or} ~ 0<\xi_i(\bb^*)\leq \frac\delta2, \\
	&&\xi_i(\bb^*+\tth/\sqrt{n})\leq0 ~ \text{or} ~ \xi_i(\bb^*)>\frac\delta2,\xi_i(\bb^*+\tth/\sqrt{n})\leq\frac\delta2 ~ \text{or} ~ 0<\xi_i(\bb^*)\leq \frac\delta2,\xi_i(\bb^*+\tth/\sqrt{n})>\frac\delta2\}, \\
	s_{4i}&=&\{(\xi_i(\bb^*),\xi_i(\bb^*+\tth/\sqrt{n})): 0<\xi_i(\bb^*)\leq \frac\delta2,\frac\delta2<\xi_i(\bb^*+\tth/\sqrt{n})\leq \frac\delta2\}.
	\ese
	Since $0<w_i<1$, applying Lemma \ref{lemma2}, we set $L_1(u)=w_iL_{\delta1}(\xi_i(\bb^*+\tth/\sqrt{n}))=w_iL_1(\xi_i(\bb^*+\tth/\sqrt{n})/\delta)$ and $L_1(v)=w_iL_{\delta1}(\xi_i(\bb^*))=w_iL_2(\xi_i(\bb^*)/\delta)$ to verify inequalities (\ref{ineR1}); set $L_2(u)=w_iL_{\delta2}(\xi_i(\bb^*+\tth/\sqrt{n}))=w_iL_2(\xi_i(\bb^*+\tth/\sqrt{n})/\delta)$ and $L_2(v)=w_iL_{\delta1}(\xi_i2(\bb^*))=w_iL_2(\xi_i(\bb^*)/\delta)$ to verify (\ref{ineR2}).
	Now we need to show that $\sumi(R_{i,n}(\tth)-ER_{i,n}(\tth))=o_\mathbb{P}(1)$:
	\be
	& &E\left\{\sumi\left(R_{i,n}(\tth)-ER_{i,n}(\tth)\right)\right\}^2, \nonumber \\
	&= & E\left\{\sumi\left(R_{i,n}(\tth)-ER_{i,n}(\tth)\right)^2-\sum_{i\neq j}\left(R_{i,n}(\tth)-ER_{i,n}(\tth)\right)\left(R_{j,n}(\tth)-ER_{j,n}(\tth)\right)\right\}.
	\label{crossproduct3}
	\ee
	The cross-product terms in (\ref{crossproduct3}) can be cancelled out, thus we have
	\allowdisplaybreaks
	\bse
	&&\sumi E\left(|R_{i,n}(\tth)-ER_{i,n}(\tth)|^2\right)=\sumi E(R_{i,n}(\tth)^2)-n(ER_{i,n}(\tth))^2, \\
	&\leq & \sumi E(R_{i,n}(\tth)^2)= \sumi E(R_{1i,n}(\tth)-R_{2i,n}(\tth))^2.
	\ese
	We can prove that the upper bound of $\sumi E(R_{1i,n}(\tth)-R_{2i,n}(\tth))^2$ tends to be 0. Since $R_{1i,n}(\tth)$ and $R_{2i,n}(\tth)$ are both non-negative, we have
	\bse
	& &\sumi E\left(R_{1i,n}(\tth)-R_{2i,n}(\tth)\right)^2\leq \sumi E(R_{1i,n}(\tth)^2)+\sumi E(R_{2i,n}(\tth)^2).
	\ese
	The upper bounds of $R_{1i,n}(\tth)$ and $R_{2i,n}(\tth)$ are given by Lemma \ref{lemma2}, so we can obtain
	\bse
	R_{1i,n}^2(\tth)&\leq& k^2|\delta^{-1}\z_i\trans\tth/\sqrt{n}|^2\ind_{\{(\xi_i(\bb^*),\xi_i(\bb^*+\tth/\sqrt{n}))\in s_{1i}\}}+M^2|\delta^{-1}\z_i\trans\tth/\sqrt{n}|^4\ind_{\{(\xi_i(\bb^*),\xi_i(\bb^*+\tth/\sqrt{n}))\in s_{2i}\}}, \\ &\leq&k^2|\delta^{-1}\z_i\trans\tth/\sqrt{n}|^2\ind_{\{|\delta^{-1}\z_i\trans\tth/\sqrt{n}|\geq|\xi_i(\bb^*)-\frac\delta2| ~\text{or} ~|\delta^{-1}\z_i\trans\tth/\sqrt{n}|\geq|\delta-\xi_i(\bb^*)|\}}+M^2|\delta^{-1}\z_i\trans\tth/\sqrt{n}|^4, \\ &=&k^2|\delta^{-1}\z_i\trans\tth/\sqrt{n}|^2U_1(|\delta^{-1}\z_i\trans\tth/\sqrt{n}|)+M^2|\delta^{-1}\z_i\trans\tth/\sqrt{n}|^4, \\
	R_{2i,n}^2(\tth)&\leq& k^2|\delta^{-1}\z_i\trans\tth/\sqrt{n}|^2\ind_{\{(\xi_i(\bb^*),\xi_i(\bb^*+\tth/\sqrt{n}))\in s_{3i}\}}+M^2|\delta^{-1}\z_i\trans\tth/\sqrt{n}|^4\ind_{\{(\xi_i(\bb^*),\xi_i(\bb^*+\tth/\sqrt{n}))\in s_{4i}\}}, \\ &\leq&k^2|\delta^{-1}\z_i\trans\tth/\sqrt{n}|^2\ind_{\{|\delta^{-1}\z_i\trans\tth/\sqrt{n}|\geq|\xi_i(\bb^*)| ~\text{or} ~ |\delta^{-1}\z_i\trans\tth/\sqrt{n}|\geq|\frac\delta2-\xi_i(\bb^*)|\}}+M^2|\delta^{-1}\z_i\trans\tth/\sqrt{n}|^4, \\ &=&k^2|\delta^{-1}\z_i\trans\tth/\sqrt{n}|^2U_2(|\delta^{-1}\z_i\trans\tth/\sqrt{n}|)+M^2|\delta^{-1}\z_i\trans\tth/\sqrt{n}|^4,
	\ese
	where $U_1(t)=\ind_{\{t\geq\min(|\xi_i(\bb^*)-\delta/2|,|\delta-\xi_i(\bb^*)|)\}}$
	and $U_2(t)=\ind_{\{t\geq\min(|\xi_i(\bb^*)|,|\delta/2-\xi_i(\bb^*)|)\}}$, for $t\in \mathcal{R}$.
	Hence,
	\allowdisplaybreaks
	\bse
	&&\sumi E(R_{1i,n}(\tth)^2)+\sumi E(R_{2i,n}(\tth)^2), \\
	&\leq & \sumi E\left\{ k^2|\delta^{-1}\z_i\trans\tth/\sqrt{n}|^2U_1(|\delta^{-1}\z_i\trans\tth/\sqrt{n}|)+M^2|\delta^{-1}\z_i\trans\tth/\sqrt{n}|^4\right\} \\
	& &+\sumi E\left\{ k^2|\delta^{-1}\z_i\trans\tth/\sqrt{n}|^2U_2(|\delta^{-1}\z_i\trans\tth/\sqrt{n}|)+M^2|\delta^{-1}\z_i\trans\tth/\sqrt{n}|^4\right\}, \\
	&\leq & \sumi E\left\{ k^2n^{-1}\delta^{-2}||\z_i||^2||\tth||^2U_1(\delta^{-1}||\z_i|| ||\tth||/\sqrt{n})+M^2\delta^{-4}n^{-2}||\z_i||^4||\tth||^4\right\} \\
	& &+\sumi E\left\{ k^2n^{-1}\delta^{-2}||\z_i||^2||\tth||^2U_2(\delta^{-1}||\z_i|| ||\tth||/\sqrt{n})+M^2\delta^{-4}n^{-2}||\z_i||^4||\tth||^4\right\} \\
	&= &  k^2||\tth||^2\delta^{-2}E\left\{ ||\Z||^2\left(U_1(\delta^{-1}||\Z|| ||\tth||/\sqrt{n})+U_2(\delta^{-1}||\Z|| ||\tth||/\sqrt{n})\right)\right\} +2M^2||\tth||^4\delta^{-4}n^{-1}E||\Z||^4.
	\ese
	(C2) implies that $E(||\Z||^2)<\infty$ and $E(||\Z||^4)<\infty$. Hence, for any $\epsilon_1$, 
	as $\delta\rightarrow0$, there exists $C$ such that $E\{||\Z||^2\ind_{\{||\Z||>C\}}\}<\delta^2\epsilon_1/2$. Therefore,
	\bse
	E\left\{||\Z||^2U_1(\delta^{-1}||\Z|| ||\tth||/\sqrt{n})\right\} &\leq&E\{||\Z||^2\ind_{\{||\Z||>C\}}\}+C^2P\left(U_1(\delta^{-1}||C|| ||\tth||/\sqrt{n})\right), \\
	E\left\{||\Z||^2U_2(\delta^{-1}||\Z|| ||\tth||/\sqrt{n})\right\} &\leq&E\{||\Z||^2\ind_{\{||\Z||>C\}}\}+C^2P\left(U_2(\delta^{-1}||C|| ||\tth||/\sqrt{n})\right).
	\ese
	By (C2), 
	the distribution of $\Z$ is not degenerate, which in turn implies that $\lim_{t\rightarrow0}P(U_{i}(t))=0$ for $i=1,2$. There exists large $N_1$ and $N_2$ such that
	\bse
	P\left(U_1(\delta^{-1}||\Z|| ||\tth||/\sqrt{n})\right)<\delta^2\epsilon_1/(2C^2) ~ \text{for} ~ n\geq N_1, \\
	P\left(U_2(\delta^{-1}||\Z|| ||\tth||/\sqrt{n})\right)<\delta^2\epsilon_1/(2C^2) ~ \text{for} ~ n\geq N_2.
	\ese
	Together with (C3), we can obtain $\delta^{-4}n^{-1}E||\Z||^4\rightarrow0$ for $n\rightarrow\infty$. As a result, for $n\geq\max(N_1,N_2)$, it is easy to get
	\bsq
	\sumi E(|R_{i,n}(\tth)-ER_{i,n}(\tth)|^2)\rightarrow 0 ~\text{as}~ n\rightarrow\infty.
	\esq
	For each fixed $\tth$, $\Lambda_n(\tth)=\frac12\tth\trans \H(\bb^*)\tth+\W_n\trans\tth/\sqrt{n}+o_p(1)$.
	Let $\ett_n=-\H(\bb^*)^{-1}\W_n/\sqrt{n}$. By Convexity Lemma in \citep{pollard1991asymptotics}, which was proved in \citep{andersen1982cox}, we have
	\bsq
	\Lambda_n(\tth)=\frac12(\tth-\ett_n)\trans \H(\bb^*)(\tth-\ett_n)-\frac12\ett_n\trans \H(\bb^*)\ett_n+r_n(\tth),
	\esq
	where for each compact set $K$ in $\mathbb{R}^{p}$,
	$\sup_{\tth\in K}|r_n(\tth)|\rightarrow0 ~\text{in probability}$.
	Since $\ett_n$ converges in distribution, there exists a compact set $K$ containing $B_\varepsilon$, where $B_\varepsilon$ is a closed ball with center $\ett_n$ and radius $\varepsilon$ with probability arbritrarily close to one. Hence,
	\bsq
	\Delta_n=\sup_{\tth\in B_\varepsilon}|r_n(\tth)|\rightarrow0 ~ \text{in probability}.
	\esq
	Now consider the behavior of $\Lambda_n$ outside $B_\varepsilon$. Suppose $\tth=\ett_n+\zeta \v$, with $\zeta>\varepsilon$ and $\v$ is a unit vector. Define $\tth^*=\ett_n+\varepsilon \v$ as the boundary point of $B_\varepsilon$. Convexity of $\Lambda_n$ and the definition of $\Delta_n$ imply
	\be \frac\varepsilon\zeta\Lambda_n(\tth)+(1-\frac\varepsilon\zeta)\Lambda_n(\ett_n) &\geq & \Lambda_n(\tth^*) \nonumber \\
	&\geq & \frac12(\tth^*-\ett_n)\trans \H(\bb^*)(\tth^*-\ett_n)-\frac12\ett_n\trans \H(\bb^*)\ett_n-\Delta_n \nonumber \\
	&\geq & \frac{C^\prime}2\varepsilon^2+\Lambda_n(\ett_n)-2\Delta_n.
	\label{outside3}
	\ee
	The right side of the inequality (\ref{outside3}) does not depend on $\tth$. Thus, we have
	\bsq
	\inf_{||\tth-\ett_n||\geq\varepsilon} \Lambda_n(\tth)\geq \Lambda_n(\ett_n)+\frac\zeta\varepsilon(\frac{C^\prime}{2}\varepsilon^2-2\Delta_n).
	\esq
	With $\pr(2\Delta_n<C^\prime\varepsilon^2/2)$ tending to one, the minmium of $\Lambda_n$ cannot occur at any $\tth$ with $||\tth-\ett_n||>\varepsilon$. Hence, for each $\varepsilon>0$ and $\wh\tth_{\lambda}=\sqrt{n}(\wh\bb_{\lambda}-\bb^*)$, we have
	\bsq
	\pr(||\wh\tth_{\lambda}-\ett_n||<\varepsilon)\longrightarrow1.
	\esq
	Therefore, we can prove that $\sqrt{n}(\wh\bb_{\lambda}-\bb^*)$ has an asymptotical normal distribution. This completes the proof.
\end{proof}

\bibliographystyle{agsm}
\bibliography{reference}
\end{document}